\documentclass[a4paper, 11pt]{article}

\usepackage{amsmath}
\usepackage{amsthm}
\usepackage{amssymb}
\usepackage{a4wide}
\usepackage{graphicx}
\usepackage{url}
\usepackage{tikz}
\usepackage{tikz-network}
\usetikzlibrary{topaths,calc}
\usepackage{subfig}
\usepackage{caption}
\usepackage{authblk}
\usepackage{complexity}

\newtheorem{theorem}{Theorem}[section]
\newtheorem{lemma}[theorem]{Lemma}
\newtheorem{claim}[theorem]{Claim}

\newtheorem{corollary}[theorem]{Corollary}
\newtheorem{remark}[theorem]{Remark}

\tikzstyle{v}=[circle,inner sep=0, minimum size =6 pt, line width = 1pt, draw=black, fill=black, text= white]

\tikzstyle{e}=[draw, line width = 1.2pt, draw=black]

\tikzstyle{V}=[circle,inner sep=0, minimum size =30 pt, line width = 1pt, draw=black, fill=white, text= black]

\tikzstyle{VC}=[circle,inner sep=0, minimum size =50 pt, line width = 1pt, draw=black, fill=white, text= black]

\newcommand{\qedclaim}{\hfill $\diamond$ \medskip}
\newenvironment{proofclaim}{\noindent{\em Proof of the claim.}}{\qedclaim}

\newcommand{\unifpos}{{\sc Uniform~POS~CNF~6}}

\newcommand{\keywords}[1]{\textbf{Keywords:} #1}

\title{Complexity of Maker-Breaker Games on Edge Sets of Graphs\thanks{This work was partially supported by the National Research Agency (ANR) project P-GASE (ANR-21-CE48-0001). Declarations of interest: none.}}

\date{}

\author[1]{Eric Duch\^{e}ne}
\author[2]{Valentin Gledel \thanks{Supported by the Kempe Foundation Grant No.~JCK-2022.1 (Sweden).}}
\author[3]{Fionn {Mc Inerney}
\thanks{{\bf Corresponding author}. Email address: fmcinern@gmail.com. Postal address: Technische Universit\"{a}t Wien, Institute of Logic and Computation, Favoritenstra{\ss}e 9-11, E192-01, 1040 Wien, Austria. Supported by the Austrian Science Fund (FWF, project Y1329).}}
\author[4]{Nicolas Nisse}
\author[1]{Nacim Oijid}
\author[1]{Aline Parreau}
\author[5]{Milo\v{s} Stojakovi\'c \thanks{Partly supported by Ministry of Science, Technological Development and Innovation of Republic of Serbia
(Grants 451-03-66/2024-03/200125 \& 451-03-65/2024-03/200125). Partly supported by Provincial Secretariat for Higher Education and Scientific Research, Province of Vojvodina (Grant No.~142-451-2686/2021).}}

\affil[1]{Univ Lyon, CNRS, INSA Lyon, UCBL, Centrale Lyon, Univ Lyon 2, LIRIS, UMR5205, F-69622 Villeurbanne, France}

\affil[2]{Universit\'{e} Savoie Mont Blanc, CNRS, LAMA, F-73000, Chamb\'{e}ry, France}

\affil[3]{Algorithms and Complexity Group, TU Wien, Austria}

\affil[4]{Universit\'{e} C\^{o}te d'Azur, Inria, CNRS, I3S, France}

\affil[5]{Department of Mathematics and Informatics, Faculty of Sciences, University of Novi~Sad, Novi Sad, Serbia}

\begin{document}
\maketitle

\begin{abstract}
We study the algorithmic complexity of Maker-Breaker games played on the edge sets of general graphs. We mainly consider the perfect matching game and the $H$-game. Maker wins if she claims the edges of a perfect matching in the first, and a copy of a fixed graph~$H$ in the second. We prove that deciding who wins the perfect matching game and the $H$-game is \PSPACE-complete, even for the latter in small-diameter graphs if $H$ is a tree. Toward finding the smallest graph $H$ for which the $H$-game is \PSPACE-complete, we also prove that such an $H$ of order~51 and size~57 exists. 

We then give several positive results for the $H$-game. As the $H$-game is already \PSPACE-complete when $H$ is a tree, we mainly consider the case where $H$ belongs to a subclass of trees.
In particular, we design two linear-time algorithms, both based on structural characterizations, to decide the winners of the $P_4$-game in general graphs and the $K_{1,\ell}$-game in trees.
Then, we prove that the $K_{1,\ell}$-game in any graph, and the $H$-game in trees are both \FPT\ parameterized by the length of the game, notably adding to the short list of games with this property, which is of independent interest. 

Another natural direction to take is to consider the $H$-game when $H$ is a cycle. While we were unable to resolve this case, we prove that the related arboricity-$k$ game is polynomial-time solvable. In particular, when $k=2$, Maker wins this game if she claims the edges of any cycle. 
\end{abstract}


\section{Introduction}

Positional games are a class of combinatorial games that received a lot of attention recently, with the books~\cite{beck2008combinatorial} and~\cite{hefetz2014positional} giving an overview of the field. Given a finite set $X$ and a family $\mathcal F$ of subsets of $X$, two players alternate claiming unclaimed elements of $X$ until all the elements are claimed. The set $X$ is referred to as the \emph{board}, and the elements of $\mathcal F$ as the \emph{winning sets}. The pair $(X,\mathcal F)$ is often called the \emph{game hypergraph}. 

As for the rules for determining the winner of the game, there are several conventions -- Maker-Breaker games, Avoider-Enforcer games, strong making games, strong avoiding games, {\it etc.} The most studied of them are Maker-Breaker games, first introduced by Erd\H{o}s and Selfridge~\cite{Erdos-selfridge}, where the players are Maker and Breaker. Maker wins the game if she claims all the elements of a set from $\mathcal F$, and Breaker wins otherwise. Unless stated otherwise, it is assumed that Maker plays first and that both players play optimally, {\it i.e.}, they play to win.

There is a number of results on Maker-Breaker games where the board of the game is the edge set $E(G)$ of a graph $G$. For the majority of the games that appear in the literature, the winning sets are all representatives of some predetermined graph structure. The four most prominent examples of the well-studied games of this kind are the perfect matching game, the $H$-game, the connectivity game, and the Hamiltonicity game, where the winning sets are all, respectively, perfect matchings, copies of a fixed graph $H$, spanning trees, and Hamiltonian cycles.

Maker-Breaker games played specifically on graphs were first studied by Chv\'atal and Erd\H{o}s~\cite{chvatal1978biased}, who looked at the $H$-game, the connectivity game, and the Hamiltonicity game played on the edges of the complete graph, and also in the biased setting, where Breaker may claim more than one edge per round, but Maker still only claims one edge per round. These games, along with the perfect matching game, were further studied in a series of papers, with major breakthroughs in~\cite{beck1982remarks}, \cite{bednarska2000biased}, \cite{gebauer2009asymptotic}, and \cite{krivelevich2011critical}, resulting in the finding of the leading term of the threshold bias (the maximum number of edges Breaker may claim each round such that Maker still wins when claiming one edge per round) for all four games more than three decades later.

In fast Maker-Breaker games, Maker's goal is to win as fast as possible, that is, in the fewest number of moves. The length (or duration) of a game in this setting was first studied by Hefetz {\it et al.}~\cite{hefetz2009fast} for the games of connectivity, perfect matching, and Hamiltonicity, with further improvements and generalizations in~\cite{hefetz2009two, mikalavcki2018fast}. These results also have direct implications on the outcomes of strong making games for this same set of games, see, {\it e.g.},~\cite{ferber2011winning}.

Another line of research was initiated by Stojakovi{\'c} and Szab{\'o} in~\cite{stojakovic2005positional}, in which they studied Maker-Breaker games played on the edge set of the Erd\H{o}s-R\'enyi random graph model, where the main question raised and partially answered was to determine the threshold probabilities for Maker's win for the games of connectivity, perfect matching, Hamiltonicity, and the $H$-game. This question was further pursued in a number of papers~\cite{hefetz2009sharp, muller2014threshold, nenadov2016threshold}.

A general overview of results on these and other games on graphs can be found in~\cite{hefetz2014positional}.

\subsection{Complexity of positional games} 

Given a positional game, we want to study the complexity of determining the winner. Schaefer was the first to show that determining the winner of a Maker-Breaker game (appearing under the name POS~CNF) with winning sets of size at most 11 is \PSPACE-complete~\cite{schaefer1978complexity}. Recently, Rahman and Watson reduced the required winning set size to six~\cite{rahman2021}. Shortly thereafter, Galliot {\it et al.} surprisingly proved that the outcome of any Maker-Breaker game with winning sets of size at most three can be determined in polynomial time~\cite{Galliot22}. This, for example, implies that the outcome of the $H$-game, for any $H$ such that $|E(H)|\leq 3$ ({\it e.g.},~the triangle game), can be decided in polynomial time. Another game that can be resolved in polynomial time using known results is the connectivity game. Indeed, Maker wins the connectivity game (playing second) in a graph $G$ if and only if $G$ contains two edge-disjoint spanning trees~\cite{lehman1964solution}, and this graph property can be checked in polynomial time~\cite{tarjan1976edge}. Regarding other conventions, Maker-Maker (a.k.a.~``the strong games'') and Avoider-Enforcer are \PSPACE-complete~\cite{byskov2004maker,Gledel2023}, and Chooser-Picker is \NP-hard~\cite{Csernenszky2012}. Lastly, the complexity of games played on the edge sets of general graphs under some of the conventions above has been considered in~\cite{Slany00}, but only in the case where some of the edges have already been claimed as part of the input. It is worth noting that the complexity of the $H$-game without pre-claimed edges is left as an open problem in~\cite{Slany00}.
 
\subsection{Complexity of short games}
 
From a parameterized complexity perspective, a natural parameter for games is the length (or duration) $k$ of a game, {\it i.e.}, the game ends after $k$ rounds. For example, in Maker-Breaker games, that effectively means that the game ends after Maker's $k^{th}$ move. This topic is of general interest, with Downey and Fellows dedicating an entire section to it, that they refer to as ``short'' or $k$-move games, in their famous textbook on parameterized complexity~\cite{downeyfellows}, where they suggested that games in general should be \AW[*]-complete parameterized by $k$. This claim is supported by the fact that chess~\cite{scott08}, geography~\cite{ADF95}, and particular pursuit-evasion games~\cite{ST10} are \AW[*]-complete parameterized by $k$. With respect to positional games, Bonnet {\it et al.}~\cite{Bonnet2017} recently proved that, when parameterized by $k$, Maker-Breaker games are \W[1]-complete, Maker-Maker games are \AW[*]-complete, and Avoider-Enforcer games are \co-\W[1]-complete. However, when the winning sets are more structured, some Maker-Breaker games are fixed-parameter tractable (\FPT) parameterized by $k$, such as Hex played on a hexagonal grid~\cite{BJS16}, and a generalization of Tic-Tac-Toe called $k$-Connect~\cite{Bonnet2017}. Although, in general, there are very few games known to be \FPT\ parameterized by $k$.

\subsection{Our contributions}

The algorithmic complexity of Maker-Breaker games played on vertex sets of graphs has been studied in general graphs (see, {\it e.g.},~\cite{BFM23,Bonnet2017,BJS16,DGPR18,ET76}), and, as seen above, there is a large literature on those played on edge sets of graphs. Surprisingly, the algorithmic complexity of the latter has only been considered in restricted graph classes like complete graphs and random graphs, or in the case where some of the edges of the graph have already been claimed~\cite{Slany00}. We study the algorithmic complexity of Maker-Breaker games played on the edge sets of general graphs. We mainly consider two of the more prominent such games: the perfect matching game and the $H$-game. In Section~\ref{sec:hardness}, we prove that it is \PSPACE-complete to decide the outcomes of the perfect matching game and the $H$-game, even for the latter in graphs of diameter at most~$6$ when $H$ is a tree (see Figure~\ref{f-H-tree} for an illustration of $H$ in this case). Toward finding the smallest graph $H$ for which the $H$-game is \PSPACE-complete, we also prove that there exists such a graph $H$ of order 51 and size 57 (see Figure~\ref{f-H} for a depiction of $H$ in this case).

We then give several positive results for the $H$-game. Since the $H$-game is already \PSPACE-complete when $H$ is a tree, we first mainly focus on the natural direction of considering the case where $H$ belongs to a subclass of trees. Specifically, in Section~\ref{sec:poly}, we design two linear-time algorithms that are both based on structural characterizations, to decide the winners of the $P_4$-game in general graphs and the $K_{1,\ell}$-game in trees. Our algorithm for the $P_4$-game improves on the algorithm of Galliot {\it et al.}~\cite{Galliot22} for this game\footnote{Recall that Galliot {\it et al.}~\cite{Galliot22} proved that the outcome of any Maker-Breaker game with winning sets of size at most $3$ can be determined in polynomial time. This is the case for the $P_4$-game since $|E(P_4)|=3$.} since their algorithm does not give any structural insight about the graphs in which Maker (Breaker, resp.) wins, and it runs in time $O(|V(\mathcal{H})|^5|E(\mathcal{H})|^2+|V(\mathcal{H})|^6\Delta(\mathcal{H}))$, where $\mathcal{H}$ is the game hypergraph. 

While we were unable to extend our result for the $P_4$-game to even the restricted case of the $P_5$-game in trees, we manage to extend our positive result for the $K_{1,\ell}$-game by considering parameterizations by the length of the game. That is, in Section~\ref{sec:FPT}, we prove that the $K_{1,\ell}$-game in any graph, and the $H$-game in trees are both \FPT\ parameterized by the length of the game. The first result relies on the fact that Maker trivially wins the $K_{1,\ell}$-game if the maximum degree of the graph is at least $2\ell-1$ (Lemma~\ref{lem:highdegree}), and thus, the $K_{1,\ell}$-game needs only to be considered in graphs of bounded maximum degree, while the second result is dependent on the structure of trees. However, most importantly, both results crucially rely on Theorem~\ref{thm:balls}, which states that the $H$-game needs only to be considered in the ball (of edges) of bounded diameter in the length of the game centered at the first edge claimed by Maker, which is interesting on its own since it could lead to other positive results for the $H$-game, and this technique could be generalized to other games on graphs. This is in stark contrast with our result that the $H$-game is \PSPACE-complete in graphs of diameter at most~$6$, even when $H$ is a tree. Furthermore, our result for the $K_{1,\ell}$-game is notably of independent interest since we add to the short list of games that are \FPT\ parameterized by the length of the game, as discussed above. 

We conclude with directions for further work in Section~\ref{sec:conclusion}. In particular, due to our hardness result for the $H$-game, another logical direction for future work is to consider the case where $H$ is a cycle. While we were not able to resolve this case, we prove that the related arboricity-$k$ game is polynomial-time solvable. In particular, when $k=2$, Maker wins this game if she claims the edges of any cycle.

\section{Preliminaries}\label{sec:preliminaries}

Unless stated otherwise, the graphs considered in this paper are finite, simple, and undirected. The vertex and edge sets of a graph $G$ will be denoted by $V(G)$ and $E(G)$, respectively. For any graph $G$ and any two vertices $x,y\in V(G)$, $dist_G(x,y)$ is the length of a shortest path between $x$ and $y$ in $G$. A graph $H$ is a {\em subgraph} of $G$ if $V(H)\subseteq V(G)$ and $E(H)\subseteq E(G)$. Let $V'\subseteq V(G)$. The subgraph {\em induced by $V'$}, denoted by $G[V']$, is the graph with vertex set $V'$ and whose edge set consists of all the edges of $G$ between two vertices of $V'$: $E(G[V']) = \{uv\in E(G) \mid u,v \in V'\}$. If $E'$ is a set of edges, the {\em edge-induced subgraph} $G[E']$ is the subgraph of $G$ whose edge set is $E'$ and whose
vertex set consists of all ends of edges of $E'$. For any integer $\ell\geq 1$, $P_{\ell}$ is the path with $\ell-1$ edges, and $K_{1,\ell}$ is the star with $\ell$ edges.

In this paper, we consider Maker-Breaker games played on the edge set of a graph $G$. In other words, the game hypergraph is a pair  $(E(G), \mathcal F)$, where $\mathcal F\subseteq 2^{E(G)}$. For all of the games we consider, we say that Maker {\em wins} in $G$ if she has a {\em winning strategy} in $G$, that is, regardless of how Breaker plays, Maker can ensure claiming the edges of a winning set $F\in \mathcal{F}$. Similarly, Breaker wins in $G$ if he has a winning strategy in $G$, that is, regardless of how Maker plays, Breaker can ensure claiming at least one edge in each winning set $F\in \mathcal{F}$. Note that, by definition, only one of the two players may have a winning strategy for a given pair  $(E(G), \mathcal F)$.  
We mainly consider the {\em perfect matching game} and the {\em $H$-game}. In the perfect matching game, $\mathcal F$ is the set of all perfect matchings of $G$. Thus, Maker wins if she manages to claim $|V(G)|/2$ vertex-disjoint edges that cover all the vertices of the graph. Note that, during the game, Maker may claim edges incident to other edges she already claimed.
In the $H$-game, $H$ refers to another (fixed) graph and the winning sets are all copies of $H$, that is, a winning set $F\in \mathcal{F}$ is a set of edges such that $G[F]$ is isomorphic to $H$. Note that $G[F]$ is not  necessarily a vertex-induced subgraph of $G$ (there may be more edges in $G$), but its vertices must be in a one-to-one correspondence with the vertex set of $H$.

Most of our results are stated for Maker playing first, but can be extended to Breaker starting.
For example, assume we can check in polynomial time whether Maker wins playing first in a family of graphs closed under edge removal. 
If Breaker starts, then we can scan through all the possible first moves of Breaker, and treat the resulting graph with one edge removed as the base graph for a new game in which Maker starts. 
This gives a polynomial-time algorithm when Breaker starts.

When it is clear what the set of winning sets is, we will use only the graph $G$ to describe a game. We will also need to consider a {\em position} of the game where some edges have already been claimed. We will denote a position by a triplet $(G,E_M,E_B)$, where $E_M$ and $E_B$ are two disjoint sets of edges that correspond to the edges already claimed by Maker and by Breaker, respectively. Note that, in general, $|E_B|\leq |E_M| \leq |E_B|+1$.

We now state some basic definitions and facts about Maker-Breaker games that will be used throughout the paper. The next lemma expresses the fact that the addition of edges can only help Maker, whereas, the deletion of edges is favorable for Breaker.

\begin{lemma}[Folklore]\label{subgraph}
Let $G'$ be a subgraph of $G$, ${\cal F} \subseteq 2^{E(G)}$ a set of winning sets, and ${\cal F}' \subseteq {\cal F} \cap 2^{E(G')}$. If Maker wins in $(E(G'),{\cal F}')$, then she wins in $(E(G),{\cal F})$. 
\end{lemma}

The {\em union} of two games $(E(G_1),\mathcal F_1)$ and $(E(G_2),\mathcal F_2)$ is simply the game $(E(G_1)\cup E(G_2),\mathcal F_1\cup \mathcal F_2)$. Note that $E(G_1)$ and $E(G_2)$ can intersect. It is not difficult to prove the following.

\begin{lemma}[Folklore]\label{lem:union}
    Let $(E(G_1),\mathcal F_1)$ and $(E(G_2),\mathcal F_2)$ be Maker-Breaker games. If Maker wins in one of the two games, then she wins in their union $(E(G_1)\cup E(G_2),\mathcal F_1\cup \mathcal F_2)$. If $E(G_1) \cap E(G_2) = \emptyset$, then the converse is also true.
\end{lemma}

A \emph{pairing strategy} is a frequently used type of strategy in positional games. Suppose that a subset of edges is partitioned into disjoint pairs. Then, the pairing strategy dictates that, during the game, whenever their opponent claims one edge of a pair, the player responds by claiming the other one. We will use pairing strategies in several places, sometimes as ingredients in more elaborate strategies.

Erd\H{o}s and Selfridge~\cite{Erdos-selfridge} gave a criterion for Breaker to win:

\begin{theorem}[Erd\H{o}s-Selfridge Criterion \cite{Erdos-selfridge}]\label{ES-criterion}
For any Maker-Breaker game with winning sets $\mathcal{F}$, if $\sum_{F\in \mathcal F} 2^{-|F|}<\frac{1}{2}$, then Breaker wins playing second.
\end{theorem}

We end this preliminary section with a lemma that states that if two edges are ``similar'', we can assume that both players will each claim one.

\begin{lemma}\label{super lemma}
Let $G$ be a (not necessarily simple) graph and $\mathcal{F}$ the set of the winning sets of a Maker-Breaker game played in $G$.
Let $(G,E_M,E_B)$ be a position of the game and $e_1, e_2$ two unclaimed edges of $G$ such that the two sets $\{F \setminus (E_M \cup \{e_1,e_2\}) \mid e_i \in F \in \mathcal{F} \text{ and } F\cap E_B = \emptyset \}$, for $i\in \{1,2\}$, (the new winning sets for Maker after claiming $e_i$) are the same. Then, $(G, E_M, E_B)$ and $(G, E_M \cup \{e_1\}, E_B \cup \{e_2\})$ have the same outcome, regardless of who plays next.
\end{lemma}

\begin{proof}
Let $G$ be a (not necessarily simple) graph, $\mathcal{F}$ the winning sets of a Maker-Breaker game played in $G$, and $(G,E_M, E_B)$ some position.
Let $e_1, e_2$ be two unclaimed edges such that the sets $\{F \setminus (E_M \cup \{e_1,e_2\}) \mid e_i \in F\in \mathcal{F} \text{ and } F\cap E_B = \emptyset \}$ are equal for $i\in \{1,2\}$. Note that a player plays next in $(G, E_M , E_B)$ if and only if they play next in $(G, E_M \cup \{e_1\}, E_B \cup \{e_2\})$.

Suppose first that Maker wins in $(G, E_M \cup \{e_1\}, E_B \cup \{e_2\})$ and let $\mathcal{S}$ be a corresponding winning strategy for her. We define her winning strategy $\mathcal{S}'$ in $(G, E_M, E_B)$ as follows:

\begin{itemize}
    \item If it is Maker's turn, then she claims the edge $e_0$ she would have claimed by $\mathcal{S}$ in $(G, E_M \cup \{e_1\}, E_B \cup \{e_2\})$.
    
    \item If Breaker claims an edge in $\{e_1,e_2\}$, then she claims the second edge in $\{e_1,e_2\}$.
    
    \item If Breaker claims any other edge $e$, then she claims the edge $e'$ she would have claimed by $\mathcal{S}$ if Breaker had claimed $e$ in $(G, E_M \cup \{e_1\}, E_B \cup \{e_2\})$.

    \item If $e_1$ and $e_2$ are the only unclaimed edges and it is Maker's turn, then she claims $e_1$.
\end{itemize}

Since $\mathcal{S}$ is a winning strategy, there exists a set of edges $F\in \mathcal{F}$ for which Maker will claim all the edges following $\mathcal{S}$ in $(G, E_M\cup \{e_1\}, E_B\cup\{e_2\})$. If $F \cap \{e_1,e_2\} = \emptyset$, then the same set of edges $F$ are claimed following $\mathcal{S}'$ in $(G, E_M, E_B)$, and so, she wins in $(G, E_M, E_B)$. Otherwise, $F \cap \{e_1,e_2\} = \{e_1\}$, as $e_2$ is attributed to Breaker. By the construction of $\mathcal{S}'$, Maker either claimed $e_1$ or $e_2$. If she claimed $e_1$, then she claimed all the edges of $F$, and therefore, she wins in $(G, E_M, E_B)$. Otherwise, she claimed $e_2$, and by the definitions of $e_1$ and $e_2$, there exists a set $F' \in \mathcal{F}$ such that $F\setminus (E_M \cup \{e_1\})$ = $F' \setminus (E_M \cup \{e_2\})$. Thus, she claimed all the edges of $F'$ and she wins in $(G, E_M, E_B)$.

Reciprocally, suppose now that Breaker wins in $(G, E_M \cup \{e_1\}, E_B \cup \{e_2\})$ and let $\mathcal{S}$ be a corresponding winning strategy for him. We define his winning strategy $\mathcal{S}'$ in $(G, E_M, E_B)$ as follows:

\begin{itemize}
    \item If it is Breaker's turn, then he claims the edge $e_0$ he would have claimed by $\mathcal{S}$ in  $(G, E_M \cup \{e_1\}, E_B \cup \{e_2\})$.
    
    \item If Maker claims an edge in $\{e_1,e_2\}$, then he claims the second edge in $\{e_1,e_2\}$.
    
    \item If Maker claims any other edge $e$, then he claims the edge $e'$ he would have claimed by $\mathcal{S}$ if Maker had claimed $e$ in $(G, E_M \cup \{e_1\}, E_B \cup \{e_2\})$.

    \item If $e_1$ and $e_2$ are the only unclaimed edges and it is Breaker's turn, then he claims $e_1$.
\end{itemize}

Since $\mathcal{S}$ is a winning strategy, for all $F\in \mathcal{F}$, there exists an edge $e_B \in F$ that Breaker will claim following $S$ in $(G, E_M \cup \{e_1\}, E_B \cup \{e_2\})$. Since, by definition, the two sets $\{F \setminus (E_M \cup \{e_1,e_2\})  | e_i \in F\in \mathcal{F} \text{ and } F\cap E_B = \emptyset \}$ are equal for $i\in\{1,2\}$, we can suppose that $e_B\neq e_2$, as otherwise $F \cup\{e_1\} \setminus \{e_2\}$ would be a winning set claimed by Maker. Therefore, as any edge that would have been claimed by Breaker by $\mathcal{S}$ has been claimed by him by $\mathcal{S}'$, he has claimed at least one edge in each winning set, and Breaker wins in $(G, E_M, E_B)$.
\end{proof}

\section{Hardness Results}\label{sec:hardness}

In this section, we prove that it is \PSPACE-complete to decide who wins the perfect matching game and the $H$-game via reductions from \unifpos\ that is known to be \PSPACE-complete~\cite{rahman2021}. Specifically, \unifpos\ is a game played on a CNF formula $\mathcal{F}$ over a set of variables $X=\{X_1,\ldots,X_n\}$, where all the clauses contain exactly six variables, each variable appears in its positive form, and each variable appears in at least one clause. Two players, named Satisfier and Falsifier, respectively, alternate claiming a variable that has not been claimed yet, with Satisfier playing first. Variables claimed by Satisfier are set to true, and those claimed by Falsifier are set to false. Satisfier wins if and only if $\mathcal{F}$ is true once every variable has been claimed. The associated decision problem, {\it i.e.}, deciding whether Satisfier wins or not, is denoted by \unifpos. Note that, in~\cite{rahman2021}, it is also proven that \unifpos\ remains \PSPACE-complete, even if Falsifier goes first and/or the number of variables $n$ is odd, which we make use of in our reductions.

\subsection{Perfect-matching game}

We first prove that it is \PSPACE-complete to determine the outcome of the perfect matching game. To simplify the reduction, the proof will be done allowing for parallel edges, that the following lemma enables to remove, see Figure~\ref{fig:double edge}.

\begin{lemma}\label{remove double edge}
Let $G$ be any graph containing two parallel edges $e_1$ and $e_2$ connecting two of its vertices $u$ and $v$. Then, there exists a complete bipartite graph $H$, of order $O(1)$, with bipartition $(A,B)$ such that, if $e_1$ and $e_2$ are removed from $G$, and a copy of $H$ is added instead, where $u$ ($v$, respectively) is made adjacent to two vertices of $A$ ($B$, respectively), then the outcome of the perfect matching game in the resulting graph $G'$ is the same as that in $G$.
\end{lemma}

\begin{figure}
    \centering
    \includegraphics[scale=0.8]{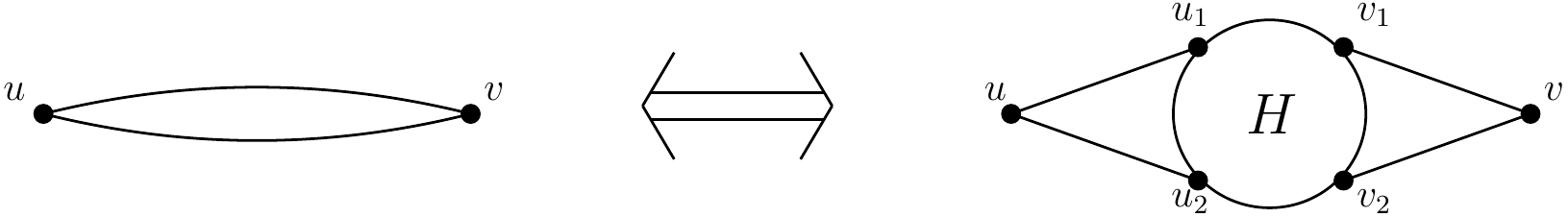}
    \caption{How parallel edges are removed in the perfect matching game in Lemma~\ref{remove double edge}.}
    \label{fig:double edge}
\end{figure}

\begin{proof}
With the parallel edges $e_1$ and $e_2$ in $G$, Maker can clearly claim the edge $uv$ ``for free'', independent of the play in the rest of the graph. That means that, playing in $G$, she has two possibilities to claim a perfect matching: 
\begin{itemize}
    \item[(i)] claim a perfect matching in $G-\{e_1, e_2 \}$, or
    \item[(ii)] claim a perfect matching in $G\setminus\{u,v\}$.
\end{itemize}

Let $H$ be a balanced complete bipartite graph with parts $A$ and $B$, and let $u_1, u_2\in A$ and $v_1, v_2\in B$. We will show that when $|A|=|B|=k$, for a suitable constant $k$ to be chosen later, Maker, playing second, has a strategy in $H$ that satisfies the following five conditions:
\begin{itemize}
    \item Maker has a perfect matching in $H$, and
    \item for every $i,j \in \{1,2\}$, Maker has a perfect matching in $H \setminus \{u_i, v_j\}$.
\end{itemize}

Once this is shown, forming $G'$ by replacing the parallel edges $e_1$ and $e_2$ in $G$ by $H$, and adding the two pairs of edges connecting $u$ to $u_1$ and $u_2$, and $v$ to $v_1$ and $v_2$, as depicted in Figure~\ref{fig:double edge}, finishes the proof. Indeed, assume Maker, playing second, has a strategy in $H$ to satisfy the above-mentioned five conditions, and, on top of that, in $G'$, she pairs $uu_1$ with $uu_2$, and $vv_1$ with $vv_2$. Then, she can claim a perfect matching in $G'$ if and only if she can satisfy (i) or (ii) in $G$, which corresponds to her being able to claim a perfect matching in $G$.

To show that Maker can satisfy the five conditions playing second in $H$, we define an auxiliary positional game in $H$. Let 
\[
{\mathcal F}_H := \left\{ E(X,Y) \mid X\subseteq A,\, Y\subseteq B,\, |X|+|Y| = k+1 \right\},
\]
and, for every $i,j \in \{1,2\}$, let
\[
{\mathcal F}_{i,j} := \left\{ E(X,Y) \mid X\subseteq A\setminus\{u_i\},\, Y\subseteq B\setminus \{v_j\},\, |X|+|Y| = k \right\}.
\]
In the auxiliary game, Maker will assume the role of Breaker (to avoid confusion, we will call this player Auxiliary Breaker), trying to claim one edge in every winning set in $\mathcal F := {\mathcal F}_H \cup \left( \cup_{i,j \in \{1,2\}} {\mathcal F}_{i,j} \right)$. If she achieves that, then Hall's condition~\cite{Hall35} for the existence of a perfect matching in a bipartite graph implies that Maker will have a claimed perfect matching in $H$, as well as in $H \setminus \{u_i, v_j\}$, for every  $i,j \in \{1,2\}$. 

To show that Auxiliary Breaker can win the auxiliary game in $H$, we apply the Erd\H{o}s-Selfridge Criterion (see Theorem \ref{ES-criterion}):
\begin{align*}
\sum_{E'\in \mathcal F} 2^{-|X|} &\leq \sum_{\ell=1}^k \binom{k}{\ell} \binom{k}{k-\ell+1} 2^{-\ell(k-\ell+1)} + 4\sum_{\ell=1}^{k-1} \binom{k-1}{\ell} \binom{k-1}{k-\ell} 2^{-\ell(k-\ell)}\\
&\leq \sum_{\ell=1}^k \binom{k}{\ell} \binom{k}{k-\ell+1} 2^{-\ell(k-\ell)} + 4\sum_{\ell=1}^{k} \binom{k}{\ell} \binom{k}{k-\ell} 2^{-\ell(k-\ell)}\\
&= \sum_{\ell=1}^k \binom{k}{\ell} \binom{k}{\ell-1} 2^{-\ell(k-\ell)} + 4\sum_{\ell=1}^{k} \binom{k}{\ell}^2 2^{-\ell(k-\ell)}\\
&\leq 2\sum_{\ell=1}^{\lceil k/2 \rceil} \binom{k}{\ell}^2 2^{-\ell(k-\ell)} + 8\sum_{\ell=1}^{\lceil k/2 \rceil} \binom{k}{\ell}^2 2^{-\ell(k-\ell)}\\
&\leq 10 \sum_{\ell=1}^{\lceil k/2\rceil} k^{2\ell} 2^{-\ell(k-\lceil k/2\rceil)}
= 10 \sum_{\ell=1}^{\lceil k/2\rceil} \left( 2^{2\log_2 k-\lfloor k/2\rfloor} \right) ^\ell.
\end{align*}

This expression tends to zero when $k$ grows, and hence, for $k$ large enough (say, $k=100$) it is less than $1/2$, and so, the Erd\H{o}s-Selfridge criterion implies Auxiliary Breaker's win.
\end{proof}

\begin{theorem}\label{thm:perfect-matching}
Deciding whether Maker wins the perfect matching game in a given graph $G$ is \PSPACE-complete.
\end{theorem}

\begin{proof}
The problem is clearly in \PSPACE\ since both the number of turns and the number of possible moves at each turn are bounded from above by $n^2$~\cite[Lemma~2.2]{schaefer1978complexity}. To prove it is \PSPACE-hard, we give a reduction from \unifpos\ where there are an odd number of variables, which, as mentioned before, is \PSPACE-hard~\cite{rahman2021}.

Let $\phi$ be an instance of \unifpos\ where there are an odd number of variables. Denote the variables in $\phi$ by $x_1, \dots, x_{2n+1}$, and the clauses in $\phi$ by $C_1,\dots,C_m$. From $\phi$, we construct the graph $G$ as follows, and recall that, by Lemma~\ref{remove double edge}, we can allow pairs of parallel edges.

\begin{itemize}
    \item For all $1\le i \le 2n+1$, introduce two vertices $v^i_{i_0}$ and $\overline{v}^i_{i_0}$, and the edge $e_i=v^i_{i_0}\overline{v}^i_{i_0}$.

    \item Then, add $n$ new vertices $a_1, \dots, a_n$, and, for all $1 \le i \le 2n+1$ and $1 \leq \ell \leq n$, add two parallel edges between $v^i_{i_0}$ and $a_{\ell}$. See Figure~\ref{fig:variables} for an illustration.
    
    \item For each clause $C_j$ in $\phi$, add a vertex $C^j$ in $G$. 
    
    \item  For each variable $x_i$ in $\phi$, let $C_{i_1}, \dots, C_{i_{k_i}}$ be the clauses containing $x_i$ in $\phi$. For all $1\leq i \leq 2n+1$ and for all $1 \le j \le k_i$, add the vertices $u^i_{i_j}, \overline{u}^i_{i_j}, v^i_{i_j}, \overline{v}^i_{i_j}$, $x^i_{i_j}$, and $y^i_{i_j}$. Also, for all $1\leq i \leq 2n+1$, add the vertices $y^i_{i_{k_i}+1}$. Then, connect them as follows for all $1 \leq i \leq 2n+1$ and $1 \le j \le k_i$ (see Figures~\ref{fig:variables in clauses}~and~\ref{fig:clauses}):
    
    \begin{itemize}
        \item Add the two edges $\overline{v}^i_{i_{j-1}}u^i_{i_j}$ and $\overline{v}^i_{i_{j-1}} \overline{u}^i_{i_j}$.
        
        \item Add two parallel edges between $u^i_{i_j}$ and $x^i_{i_j}$.
        
        \item Add two parallel edges between $x^i_{i_j}$ and $C^{i_j}$.
        
        \item Add two parallel edges between $x^i_{i_j}$ and $y^i_{i_j}$.
        
        \item Add two parallel edges between $u^i_{i_j}$ and $\overline{u}^i_{i_j}$.
        
        \item Add two parallel edges between $\overline{u}^i_{i_j}$ and $v^i_{i_j}$.
        
        \item Add two parallel edges between $v^i_{i_j}$ and $\overline{v}^i_{i_j}$.
    \end{itemize}

    Also, for all $1 \leq i \leq 2n+1$, add two parallel edges between $\overline{v}^i_{i_{k_i}}$ and $y^i_{i_{k_i}+1}$.

    \item If $|V(G)|$ is currently an odd number, then add the vertex $y^0_{0_0}$ in $G$.
        
    \item For each pair of vertices among the $y^i_{i_j}$'s (including $y^0_{0_0}$ if it exists), add two parallel edges between them (see Figure~\ref{fig:yvertices}).
\end{itemize}

\begin{figure}[ht]
    \centering
    \includegraphics{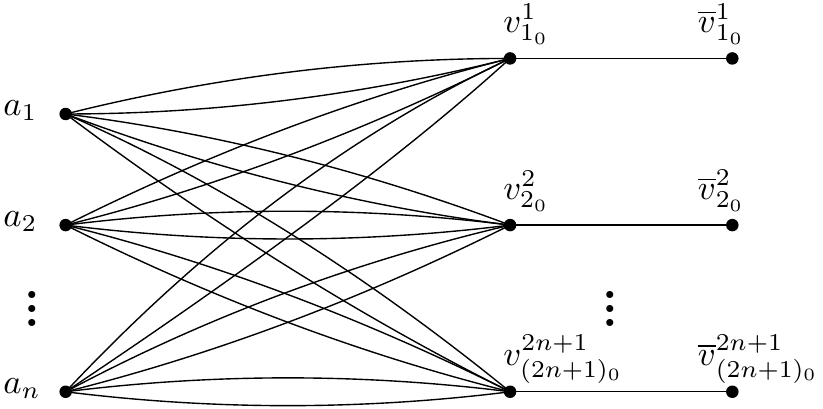}
    \caption{The variable gadget in the graph $G$ constructed in the proof of Theorem~\ref{thm:perfect-matching}.
    }
    \label{fig:variables}
\end{figure}

\begin{figure}[ht]
    \centering
    \includegraphics[scale=.575]{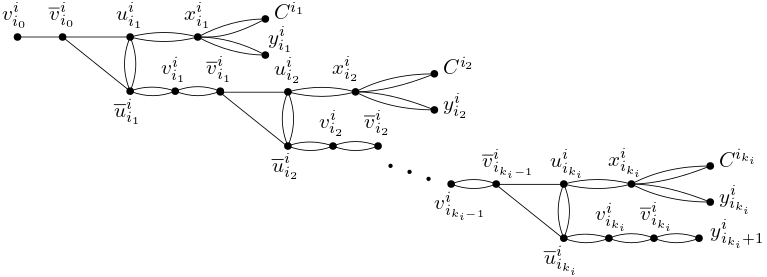}
    \caption{Construction for a variable $x_i$ in clauses $C^{i_1}, \dots, C^{i_{k_i}}$ in $\phi$ in the graph $G$ constructed in the proof of Theorem~\ref{thm:perfect-matching}.  }
    \label{fig:variables in clauses}
\end{figure}

\begin{figure}[ht]
    \centering
    \includegraphics{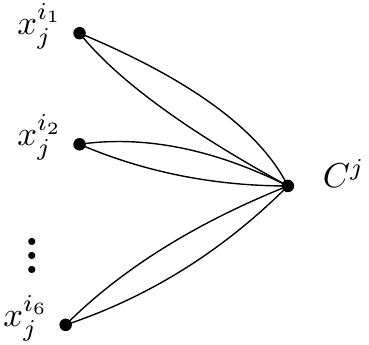}
    \caption{Construction for a clause $C_j = (x_{i_1} \vee \dots \vee x_{i_6})$ in $\phi$ in the graph $G$ constructed in the proof of Theorem~\ref{thm:perfect-matching}.}
    \label{fig:clauses}
\end{figure}

\begin{figure}[ht]
    \centering
    \includegraphics{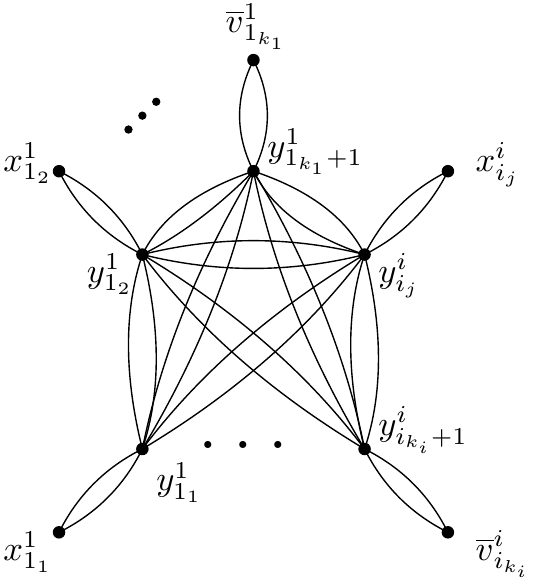}
    \caption{The $y^i_{i_j}$ vertices in the graph $G$ constructed in the proof of Theorem~\ref{thm:perfect-matching}.
    }
    \label{fig:yvertices}
\end{figure}

Note that $G$ is clearly constructed in polynomial time. We prove that Satisfier wins in $\phi$ if and only if Maker wins the perfect matching game in $G$. By Lemma~\ref{super lemma}, the outcome of the perfect matching game in $G$ is the same as the outcome in $G$, where, for each pair of parallel edges, both Maker and Breaker have claimed one of the two edges, and so, we can assume they have done so in what follows.

First, we prove the simpler of the two directions, that is, if Falsifier wins in $\phi$, then Breaker wins the perfect matching game in $G$. Assume that Falsifier has a winning strategy $\mathcal{S}$ in $\phi$. Consider the following strategy for Breaker in $G$:

\begin{itemize}
    \item If Maker claims an edge $e_i = v^i_{i_0}\overline{v}^i_{i_0}$, then Breaker answers by claiming an edge $e_j= v^j_{j_0}\overline{v}^j_{j_0}$, where $x_j$ is the variable that would have been claimed by Falsifier in $\phi$ according to $\mathcal{S}$ if Satisfier claimed $x_i$ in $\phi$.
    
    \item If Maker claims any edge other than an $e_i$ before all the $e_i$'s have been claimed, then Breaker claims an arbitrary $e_i$. Then, as there is an odd number of $e_i$'s, by pairing them, Breaker can ensure claiming at least $n+1$ of them. Thus, by construction, Maker will not be able to have a perfect matching containing all the $v^i_{i_0}$'s. Indeed, after this step, at least $n+1$ of the $e_i$'s are claimed by Breaker, and so, their respective $n+1$ $v^i_{i_0}$'s must be matched with their only remaining neighbors, the $a_{\ell}$'s, of which there are only $n$, and thus, this is not possible.
\end{itemize}

Hence, we can assume that all the edges $e_i$ have been claimed during the first $2n+1$ moves. Consider the valuation obtained in $\phi$ if all the $x_i$ variables associated to the $e_i$ edges claimed by Maker are the variables set to true, and all the $x_i$ variables associated to the $e_i$ edges claimed by Breaker are the variables set to false. By the hypothesis that $\mathcal{S}$ is a winning strategy for Falsifier in $\phi$, there exists a clause $C_z$ that is not satisfied by this valuation in~$\phi$. Let $x_{f_1}, \dots, x_{f_6}$ be the variables in $C_z$ in $\phi$. Since Breaker claimed all the edges $e_i$ corresponding to the variables $x_i$ in $\phi$ that Falsifier would have claimed according to $\mathcal{S}$, the edges $e_{f_1},\dots,e_{f_6}$ are claimed by Breaker in $G$. Recall that, as the number of variables is odd, it is Breaker's turn. Breaker plays as follows for all $\ell\in \{f_1,\dots,f_6\}$ and for $j=1$ to $k_{\ell}$ while $\ell_j\leq z$:
    \begin{itemize}
        \item  
         If $\ell_j<z$, then Breaker claims $\overline{v}^{\ell}_{\ell_{j-1}}\overline{u}^{\ell}_{\ell_j}$. As the only remaining edge available to match $\overline{v}^{\ell}_{\ell_{j-1}}$ is $\overline{v}^{\ell}_{\ell_{j-1}}u^{\ell}_{\ell_j}$, Maker has to claim it. Now, all the edges adjacent to $\overline{u}^{\ell}_{\ell_j}$ have been claimed, and Maker has claimed only two of them, of which only one does not interfere with the edges already forced in the matching: $\overline{u}^{\ell}_{\ell_j}v^{\ell}_{\ell_j}$. Thus, $\overline{u}^{\ell}_{\ell_j}v^{\ell}_{\ell_j}$ has to be in any perfect matching claimed by Maker, which forces the edge $v^{\ell}_{\ell_j}\overline{v}^{\ell}_{\ell_j}$ claimed by Maker to not be in the matching.
        
        \item If $\ell_j=z$, then Breaker claims $\overline{v}^{\ell}_{\ell_j-1}u^{\ell}_z$, forcing Maker to claim $\overline{v}^{\ell}_{\ell_j-1}\overline{u}^{\ell}_z$ to match $\overline{v}^{\ell}_{\ell_j-1}$. Now, the only edge that Maker can use to match $u^{\ell}_z$ is the edge $u^{\ell}_zx^{\ell}_z$.
    \end{itemize}
    
By the above strategy for Breaker, any perfect matching contained in the edges claimed by Maker has to contain the edges $u^{\ell}_zx^{\ell}_z$ for all $\ell\in \{f_1,\dots,f_6\}$. Therefore, all the vertices adjacent to $C^z$ are already matched, and thus, it cannot be matched, and Breaker wins.

Now, we prove that if Satisfier wins in $\phi$, then Maker wins the perfect matching game in $G$. Assume that Satisfier has a winning strategy $\mathcal{S}$ in $\phi$. We construct a strategy for Maker in $G$ as follows:

First, Maker claims the edge $e_i = v^i_{i_0}\overline{v}^i_{i_0}$ corresponding to the variable $x_i$ that Satisfier would have claimed first in $\phi$ according to $\mathcal{S}$. Then, 

\begin{itemize}
    \item If Breaker claims an edge $e_j= v^j_{j_0}\overline{v}^j_{j_0}$, then Maker claims the edge $e_i= v^i_{i_0}\overline{v}^i_{i_0}$ corresponding to the variable $x_i$ that Satisfier would have claimed according to $\mathcal{S}$ if Falsifier had claimed $x_j$ in $\phi$.
    
    \item For any variable $x_i$ in a clause $C_j$ in $\phi$, Maker pairs the edges $\overline{v}^i_{i_{j-1}}u^i_{i_j}$ and $\overline{v}^i_{i_{j-1}}\overline{u}^i_{i_j}$ in $G$, and so, if Breaker claims one of them, she claims the other one.
\end{itemize}

Note that, by construction, one of these moves is always available, as after the first move of Maker, any set where moves are considered has an even number of unclaimed edges remaining.

Suppose that Maker employs this strategy until the end of the game. Then, she will have claimed the $e_i$ edges corresponding to the $x_i$ variables claimed by Satisfier in $\phi$ according to $\mathcal{S}$. We extract a perfect matching from the edges she claimed as follows:

\begin{itemize}
    \item Add in the matching, the $n+1$ $e_i$'s claimed by Maker. The $n$ $v^i_{i_0}$'s that are not in these edges are paired with the $n$ $a_{\ell}$'s.
    
    \item Let $i_1, \dots, i_{n+1}$ be the indices of the $e_i$ edges claimed by Maker. For $\ell=i_1$ to $i_{n+1}$, add in the matching all the edges $u^{\ell}_j\overline{u}^{\ell}_j$ and $v^{\ell}_j\overline{v}^{\ell}_j$ such that $x_{\ell}$ is in $C_j$ in $\phi$.
    
    \item For each clause $C_j$, as at least one variable $x_i$ in $C_j$ in $\phi$ corresponding to an edge $x_i$ claimed by Maker is set to true, consider such an $i$, and add the edge $x^i_jC^j$ in the matching. For each other variable $x_{\ell}$ in $C_j$ in $\phi$ corresponding to an edge $x_{\ell}$ claimed by Maker ($\ell\neq i$), add the edge $x^{\ell}_jy^{\ell}_j$ in the matching.
    
    \item For each variable $x_i$ with $i \notin \{i_1, \dots, i_{n+1}\}$ being in clauses $C_{i_1}, \dots, C_{i_{k_i}}$ in $\phi$:
    
    \begin{itemize}
        \item For $j=1$ to $k_i$, while Maker has claimed $\overline{v}^i_{i_{j-1}}u^i_{i_j}$, add in the matching the edges, $\overline{v}^i_{i_{j-1}}u^i_{i_j}$, $x^i_{i_j}y^i_{i_j}$, and $\overline{u}^i_{i_j}v^i_{i_j}$. If and once Maker has claimed an edge $\overline{v}^i_{i_{z-1}}\overline{u}^i_{i_z}$, add in the matching $\overline{v}^i_{i_{z-1}}\overline{u}^i_{i_z}$, $u^i_{i_z}x^i_{i_z}$, and $v^i_{i_z}\overline{v}^i_{i_z}$, and exit the for loop. Then, for $j = z+1$ to $k_i$, add in the matching the edges $u^i_{i_j}\overline{u}^i_{i_j}$, $x^i_{i_j}y^i_{i_j}$, and $v^i_{i_j}\overline{v}^i_{i_j}$.
        
        \item If Maker has claimed no edge $\overline{v}^i_{i_{z-1}}\overline{u}^i_{i_z}$ at the end, add in the matching $\overline{v}^i_{i_{k_i}}y^i_{i_{k_i}+1}$.
    \end{itemize}
    
    \item Note that this strategy matches all the vertices except some of the $y^i_{i_j}$'s (including $y^0_{0_0}$ if it exists). As $G$ contains an even number of vertices, and every pair of remaining vertices is connected by two parallel edges, by considering any matching among the remaining $y^i_{i_j}$'s, Maker has claimed a perfect matching and wins. \qedhere
\end{itemize}
\end{proof}

\subsection{$H$-game}

We begin by proving that the $H$-game is \PSPACE-complete in graphs of small diameter when $H$ is a tree. Hence, even when $H$ is a relatively basic graph, determining the outcome of the $H$-game is hard. Furthermore, the $H$-game being hard in graphs of small diameter contrasts with a later result (Corollary~\ref{FPT-star-game}) that shows that when $H$ is a star, the $H$-game is \FPT\ parameterized by the length of the game $k$ since the diameter of the graph (after Maker's first move) can be bounded by a function of $k$. 

\begin{theorem}\label{thm:t-game}
There exists a tree $H$ such that deciding whether Maker wins the $H$-game in a given graph $G$ is \PSPACE-complete, even if $G$ has diameter at most~$6$.
\end{theorem}

\begin{proof}
We show that the statement holds for the tree $H$ in Figure~\ref{f-H-tree}.

\begin{figure}[htp]
\centering \includegraphics[width=\textwidth]{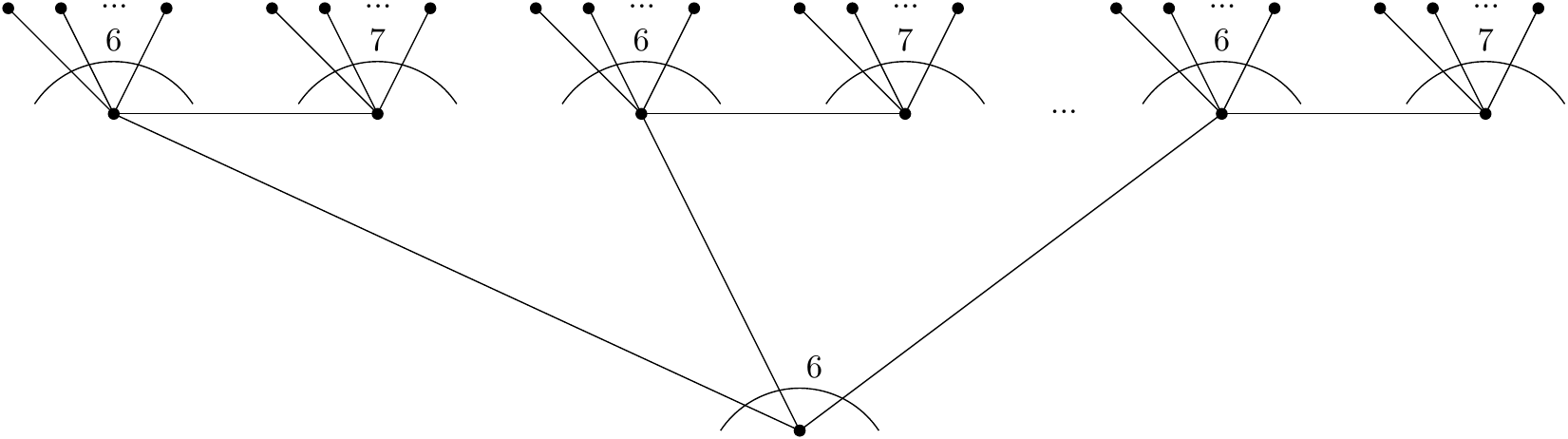}
\caption{The tree $H$ in the proof of Theorem~\ref{thm:t-game}.} \label{f-H-tree}
\end{figure}

The problem is clearly in \PSPACE\ since both the number of turns and the number of possible moves at each turn are bounded from above by $n^2$~\cite[Lemma~2.2]{schaefer1978complexity}. To prove it is \PSPACE-hard, we give a reduction from \unifpos\ where Falsifier plays first, which as mentioned before, is known to be \PSPACE-hard~\cite{rahman2021}.

\begin{figure}[htp]
\centering \includegraphics[width=\textwidth]{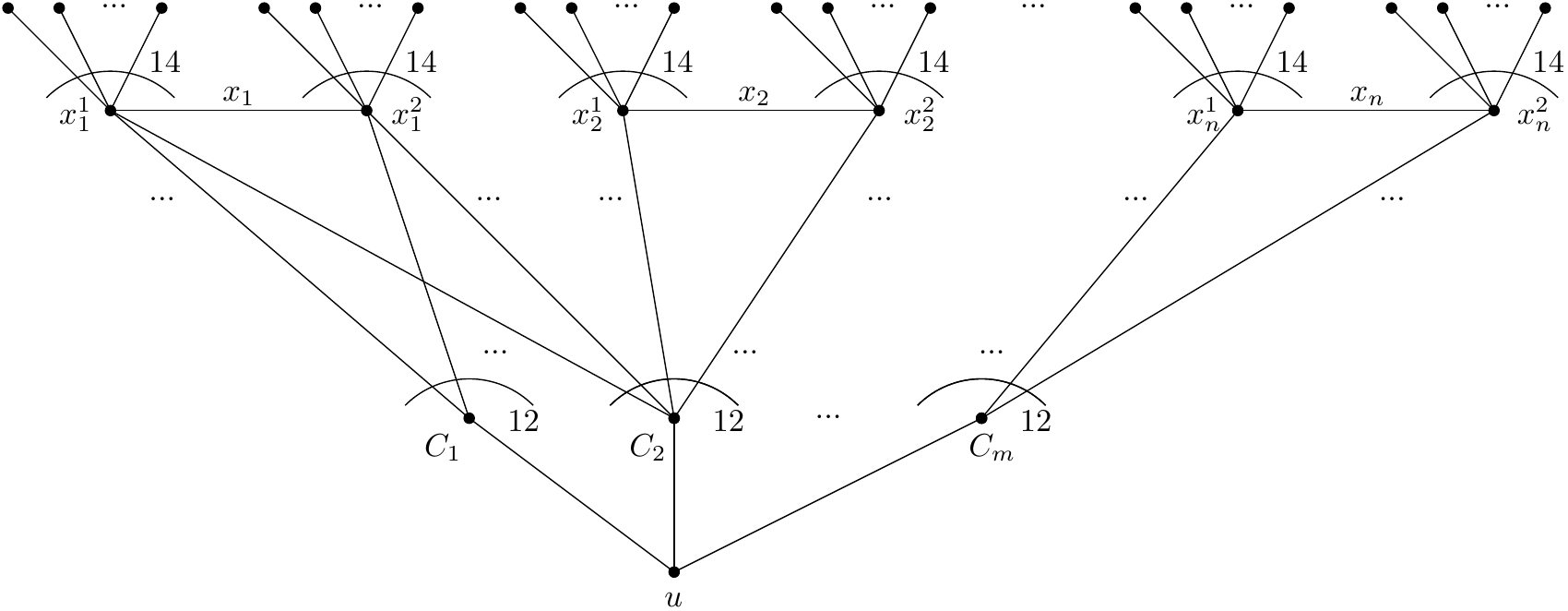}
\caption{The graph $G$ constructed in the proof of Theorem~\ref{thm:t-game}. In this example, the variable $x_1$ appears in the clauses $C_1$ and $C_2$, the variable $x_2$ appears in the clause $C_2$, and the variable $x_n$ appears in the clause $C_m$.} \label{f-construction-tree}
\end{figure}

Let $\phi$ be an instance of \unifpos\ in which Falsifier plays first. From $\phi$, we construct the graph $G$ as follows. For each clause $C_j$ in $\phi$, introduce a new clause vertex $C_j$ in $G$, and, for each variable $x_i$ in $\phi$, introduce a new variable edge $x_i=x^1_ix^2_i$ in $G$, all pairwise vertex-disjoint from each other. For all $1\leq i \leq n$ and $1\leq j \leq m$, if the variable $x_i$ is contained in the clause $C_j$ in $\phi$, then add the edges $x^1_iC_j$ and $x^2_iC_j$ in $G$. For each $1\leq i \leq n$, add 28 vertices, and make 14 of them adjacent to $x^1_i$, and the other 14 adjacent to $x^2_i$. Lastly, to ensure $G$ has diameter at most~$6$, add a vertex $u$ and, for all $1\leq j \leq m$, make it adjacent to $C_j$. Note that $G$ is constructed in polynomial time. See Figure~\ref{f-construction-tree} for an illustration of $G$.

We prove that Satisfier wins in $\phi$ (recall that Falsifier plays first) if and only if Breaker wins the $H$-game in $G$. First, we prove the following useful claims.

\begin{claim}\label{clm:oneway1}
Suppose that, for all $1\leq i\leq n$ and $1\leq j\leq m$ such that the edges $x^1_iC_j$ and $x^2_iC_j$ exist, Breaker claims at least one of $x^1_iC_j$ and $x^2_iC_j$. In that case, if Maker is to claim a copy of $H$ in $G$, then the unique vertex of degree 6 in $H$ must be a clause vertex in $G$.
\end{claim}

\begin{proofclaim}
Every other vertex in $G$ either has degree 1 or cannot be adjacent to 6 vertices that have at least 8 edges incident to each of them that Maker can claim, due to Breaker's strategy.
\end{proofclaim}

\begin{claim}\label{clm:oneway2}
Suppose that, for all $1\leq i\leq n$ and $1\leq j\leq m$ such that the edges $x^1_iC_j$ and $x^2_iC_j$ exist, Breaker claims at least one of $x^1_iC_j$ and $x^2_iC_j$. In that case, if Maker is to claim a copy of $H$ in $G$, then each of the pairs of adjacent degree-$8$ vertices in $H$ must be the two vertices of a variable edge in $G$.
\end{claim}

\begin{proofclaim}
By Claim~\ref{clm:oneway1}, the unique vertex of degree 6 in $H$ is a clause vertex. Thus, the 6 vertices of degree 8 adjacent to this vertex of degree 6 in $H$ must each be a vertex of 6 different variable edges in $G$. Indeed, $u$ cannot be one of these vertices since all of $u$'s neighbors can have at most $7$ incident edges claimed by Maker due to Breaker's strategy. Due to Breaker's strategy, for each of these vertices of degree 8, the only vertex adjacent to them that has at least 8 edges incident to it that Maker can claim, is the other vertex in each of the same variable edges.
\end{proofclaim}

First, we prove that if Satisfier wins in $\phi$, then Breaker wins the $H$-game in $G$. Assume that Satisfier wins in $\phi$. Breaker employs the following pairing strategy. If Maker claims a variable edge $x_i$, then Breaker follows his winning strategy as Satisfier in $\phi$ by claiming the variable edge in $G$ corresponding to the variable he wants to set to true in $\phi$ assuming that Maker just set the variable $x_i$ to false in $\phi$. If Maker claims an edge $x^1_iC_j$ ($x^2_iC_j$, resp.), then Breaker claims $x^2_jC_j$ ($x^1_jC_j$, resp.). If Maker claims an edge incident to a degree-1 vertex, then Breaker claims an edge incident to the same vertex of that variable edge and a degree-1 vertex (there are an even number of these edges). Lastly, if Maker claims an edge incident to $u$, then Breaker claims another edge incident to $u$. Whenever Breaker cannot employ his strategy, he claims an arbitrary edge, and then goes back to following his strategy. For a contradiction, assume that, at the end of the game, Maker claimed a copy of $H$. Then, by Claims~\ref{clm:oneway1} and~\ref{clm:oneway2}, there exists a clause such that all of the variable edges corresponding to the variables it contains in $\phi$ have been claimed by Maker. This contradicts the fact that Satisfier wins in $\phi$ since Breaker followed Satisfier's winning strategy in $\phi$ on the variable edges of $G$.

Now, we prove that if Falsifier wins in $\phi$, then Maker wins the $H$-game in $G$. Assume that Falsifier wins in $\phi$. Maker first claims a variable edge in $G$ that corresponds to the variable she wants to set to false in $\phi$ according to her winning strategy as Falsifier in $\phi$. Then, Maker employs the following pairing strategy. If Breaker claims a variable edge $x_i$, then Maker follows her winning strategy as Falsifier in $\phi$ by claiming the variable edge in $G$ corresponding to the variable she wants to set to false in $\phi$ assuming that Breaker just set the variable $x_i$ to true in $\phi$. If Breaker claims an edge $x^1_iC_j$ ($x^2_iC_j$, resp.), then Maker claims $x^2_jC_j$ ($x^1_jC_j$, resp.). If Breaker claims an edge incident to a degree-1 vertex, then Maker claims an edge incident to the same vertex of that variable edge and a degree-1 vertex (there are an even number of these edges). Lastly, if Breaker claims an edge incident to $u$, then Maker claims another edge incident to $u$. Whenever Maker cannot employ her strategy, she claims an arbitrary edge, and then goes back to following her strategy. Since Maker followed Falsifier's winning strategy in $\phi$ on the variable edges of $G$, for at least one clause, she will have claimed all of the variable edges corresponding to the variables contained in that clause in $\phi$. We can easily locate a Maker's copy of $H$ containing that clause's vertex as the unique vertex of degree 6 in $H$.
\end{proof}

As can be seen in Figure~\ref{f-H-tree}, the tree $H$ from the proof of Theorem~\ref{thm:t-game} has order~$91$. It would be interesting to know the order and/or size of the smallest graph $H$ for which the $H$-game remains \PSPACE-complete. As a step in this direction, in a more involved proof, we now show that the $H$-game is \PSPACE-complete for a graph $H$ of order~$51$ and size~$57$. It is worth noting that the orders and sizes of our graphs $H$ in our reductions are dependent on the fact that \unifpos\ is \PSPACE-hard, but it is not known whether the analogously defined {\sc Uniform~POS~CNF~5} or {\sc Uniform~POS~CNF~4} are \PSPACE-hard, which would allow for smaller $H$'s to be constructed.

\begin{theorem} \label{t:H-hard}
There exists a graph $H$ of order~$51$ and size~$57$, such that deciding whether Maker wins the $H$-game in a given graph $G$ is \PSPACE-complete.
\end{theorem}

\begin{proof}
The problem is in \PSPACE\ by the proof of Theorem~\ref{thm:t-game}. Let $H$ be the graph consisting of a triangle and six 5-cycles, pairwise vertex-disjoint, where one vertex of each 5-cycle is connected to the same vertex of the triangle by a path of length~$4$, as depicted in Figure~\ref{f-H}.

\begin{figure}[htp]
\centering \includegraphics{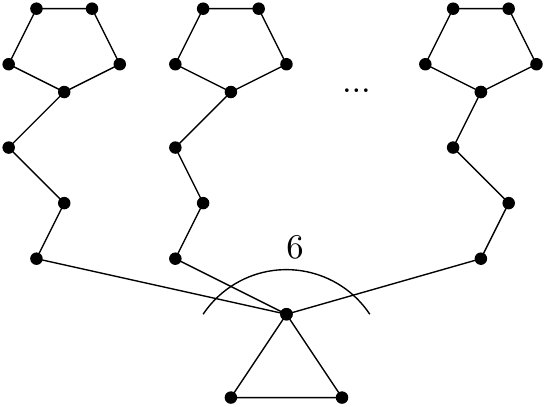}
\caption{The graph $H$ in the proof of Theorem~\ref{t:H-hard}.} \label{f-H}
\end{figure}

Before describing the reduction to prove it is \PSPACE-hard, we need to introduce some auxiliary graphs, starting with the graph $X$ (see Figure~\ref{f-X}), and prove several intermediate results.

\begin{claim} \label{c:X}
Playing on the edges of $X$ as the second player, Maker can claim a path of length~$3$ connecting $t$ to $u_1$ or $u_2$.
\end{claim}

\begin{figure}[htp]
\centering \includegraphics{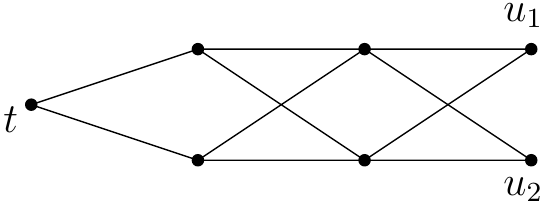}
\caption{The graph $X$ in the proof of Theorem~\ref{t:H-hard}.} \label{f-X}
\end{figure}

\begin{proofclaim}
Maker as the second player employs the following pairing strategy that clearly accomplishes her goal. For each vertex $w \in V(X)\setminus \{u_1,u_2\}$, she pairs the two edges between $w$ and its neighbors on the right.
\end{proofclaim}

The next auxiliary graph is $Y$, whose structure is shown in Figure~\ref{f-Y}.
\begin{figure}[htp]
\centering \includegraphics{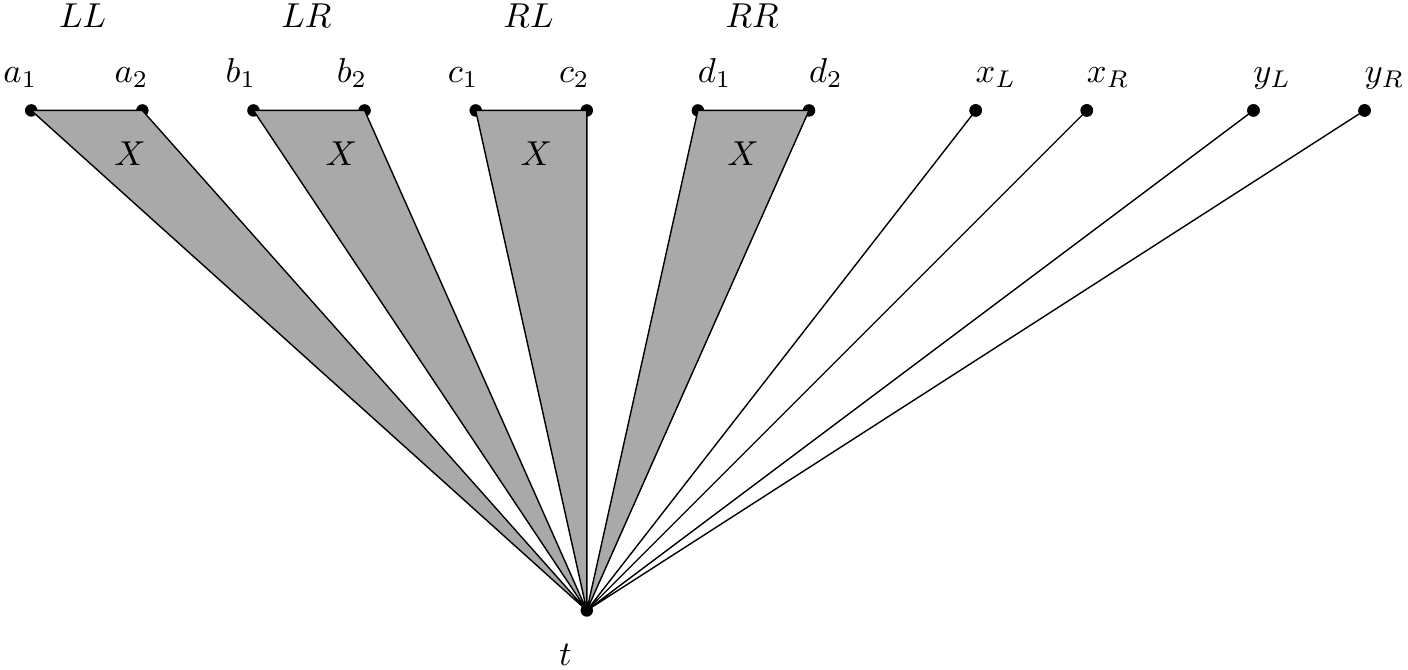}
\caption{The structure of the graph $Y$ (some edges are omitted) in the proof of Theorem~\ref{t:H-hard}.}  \label{f-Y}
\end{figure}
Note that $Y$ contains four copies of $X$ shown in gray (where the vertex $t$ in $Y$ corresponds to the identification of the vertex $t$ of each copy of $X$), as well as the following \emph{additional edges} that are not shown in the figure. Each of the vertices $a_1$ and $a_2$ (that correspond to the vertices $u_1$ and $u_2$ of the corresponding copy of $X$) is connected by two edges to the vertices $x_L$ and $y_L$ (this is indicated by the letters ``LL'' above the first copy of $X$). Similarly, each of the vertices $b_1$ and $b_2$ is connected by two edges to the vertices $x_L$ and $y_R$ (indicated by ``LR''). The vertices $c_1$ and $c_2$, and $d_1$ and $d_2$ are connected accordingly as in Figure~\ref{f-Y}.

\begin{claim} \label{c:Y}
Playing on the edges of $Y$ as the second player, Maker can claim a $C_5$ containing the vertex $t$.
\end{claim} 
\begin{proofclaim}
Maker partitions the edge set of $Y$ into several parts, and plays on each of them separately, responding in each of them as the second player. 

On each copy of $X$, Maker plays to claim a path connecting $t$ and one of the two top vertices (the two vertices corresponding to $u_1$ and $u_2$ in that copy of $X$), which is possible by Claim~\ref{c:X}. For each of the vertices $a_1,a_2,b_1,b_2,c_1,c_2,d_1,d_2$, Maker pairs the two edges going toward (two out of four of) the vertices $x_L, x_R, y_L, y_R$. Finally, she pairs the two edges between $t$ and $x_L$ and $x_R$, and the two edges between $t$ and $y_L$ and $y_R$.

Following this strategy, after all the edges are claimed, Maker will claim, without loss of generality, the edges $tx_L$ and $ty_L$. Also, in the leftmost copy of $X$, she will claim, without loss of generality, a path from $t$ to $a_1$. Now, due to pairing, one of the two edges $a_1x_L$ and $a_1y_L$ is claimed by Maker, completing a $C_5$ through $t$ that is fully claimed by her.
\end{proofclaim}

\begin{claim}
Playing on the edges of the clique $K_6$ as the second player, and given any vertex $p \in V(K_6)$, Maker can claim a triangle containing the vertex $p$.
\end{claim}

\begin{proofclaim} \label{c:K6}
For her first two moves, Maker claims two edges incident to $p$, denote them by $pv_1$ and $pv_2$. If, by Maker's next move, Breaker has not claimed $v_1v_2$, then Maker does it and completes her goal.

Otherwise, there must be another unclaimed edge incident to $p$, denote it by $pv_3$, and Maker claims it. After the following move of Breaker, the edges $v_1v_3$ and $v_2v_3$ must be claimed by him, as otherwise Maker is done. But, then there is another unclaimed edge incident to $p$, denote it by $pv_4$, and Maker claims it. By Maker's next move, Breaker cannot have claimed all the edges $v_iv_j$, $i,j\in\{1,2,3,4\}$, so Maker claims one of them in the following move completing her triangle at $p$.
\end{proofclaim}

The last auxiliary graph is the graph $Z$ as depicted in Figure~\ref{f-Z}.

\begin{claim} \label{c:Z}
Playing on the edges of $Z$ as the second player, Maker can claim a path of length~$3$ from the vertex $\ell$ to the vertex $r$.
\end{claim}
\begin{figure}[htp]
\centering \includegraphics{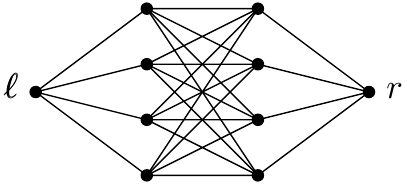}
\caption{The graph $Z$ in the proof of Theorem~\ref{t:H-hard}.} \label{f-Z}
\end{figure}

\begin{proofclaim}
In her first four moves, Maker claims two edges incident to $\ell$, and two edges incident to $r$. If, at that point, one of the four edges closing a Maker's path of length~$3$ between $\ell$ and $r$ is unclaimed, she claims it and she is done.

Otherwise, all these edges are claimed by Breaker, and so, at most one edge claimed by Breaker is incident to either $\ell$ or $r$. In her next two moves, Maker claims one more edge at both $\ell$ and $r$, and so, there are nine edges that complete Maker's path of length~$3$ from $\ell$ to $r$. As Breaker claimed only seven edges before Maker's next move, at least two of those edges are unclaimed, and Maker completes her path from $\ell$ to $r$.
\end{proofclaim}

We are now ready to describe the reduction. To prove it is \PSPACE-hard, we give a reduction from \unifpos\ where Falsifier plays first, as in the proof of Theorem~\ref{thm:t-game}.

\begin{figure}[htp]
\centering \includegraphics{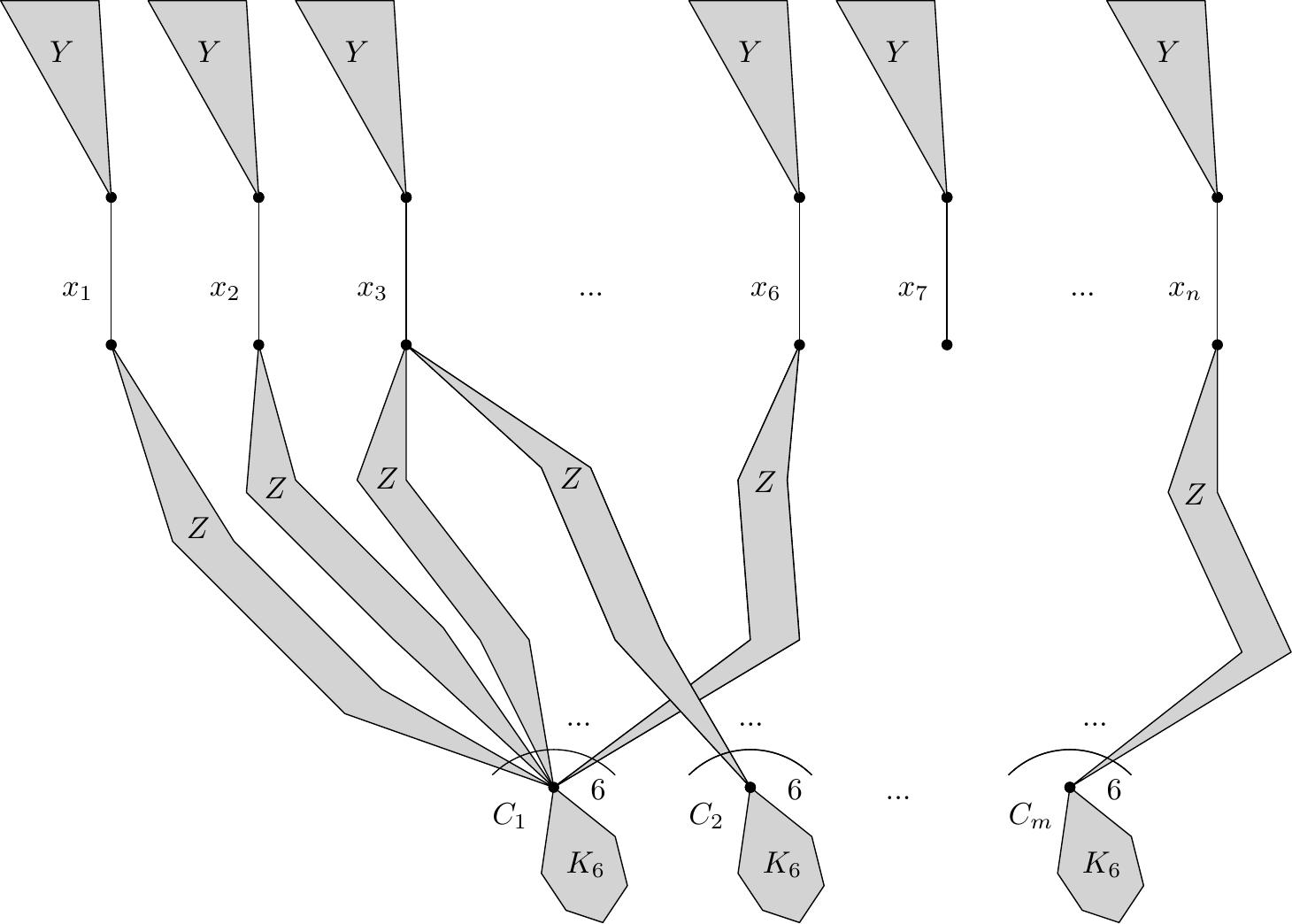}
\caption{The graph $G$ constructed in the proof of Theorem~\ref{t:H-hard}.} \label{f-construction}
\end{figure}

Let $\phi$ be an instance of \unifpos\ in which Falsifier plays first. From $\phi$, we construct the graph $G$ in the following way. For each clause $C_j$ in $\phi$, $j\in\{1,\dots, m\}$, we add a new copy of the clique $K_6$ and denote one of its vertices as $C_j$. For each variable $x_i$ in $\phi$, $i\in \{1, \dots, n\}$, we introduce a new edge denoted by $x_i$, all pairwise vertex-disjoint. For each variable edge $x_i$, we add a new copy of the graph $Y$ and identify one of the ends of $x_i$, called its ``upper'' vertex,  with the vertex $t$ of $Y$.
The other vertex of the variable edge will be called the ``lower'' vertex. 
Finally, for each variable $x_i$ and each clause $C_j$ that contains that variable in $\phi$, we introduce a new copy of the graph $Z$, identifying the vertex $\ell$ ($r$, respectively) of this copy of $Z$ with the ``lower'' vertex of $x_i$ (with $C_j$, respectively). See Figure~\ref{f-construction} for an illustration of $G$, and note that $G$ is clearly constructed in polynomial time.

We prove that Satisfier wins in $\phi$ if and only if Breaker wins the $H$-game in $G$.

First, we prove that if Satisfier wins in $\phi$, then Breaker wins the $H$-game in $G$. Assume that Satisfier wins in $\phi$. Then, Breaker can apply the winning strategy of Satisfier in $\phi$ on the variable edges of $G$, and play arbitrarily in the rest of the graph $G$. For a contradiction, assume that, at the end of the game, Maker claimed a copy of $H$. 

The graph $H$ contains a triangle, the graph $Y$ is triangle-free, and the middle part of $G$ (what remains of $G$ when all of the copies of $Y$ and $K_6$ are removed) is bipartite. Hence, the triangle in Maker's copy of $H$ needs to be within one of the copies of $K_6$ on one of the clause vertices, denote it by $\Tilde{C}$. But, in $H$, Maker also claimed six 5-cycles connected by Maker's paths of length~$4$ to that Maker's triangle. As the other $K_6$'s are too far away, the only remaining non-bipartite parts of the base graph that are close enough are the six copies of $Y$ glued to the edges corresponding to the variables in the clause~$\Tilde{C}$ in $\phi$. Since Breaker wins as Satisfier in $\phi$ on the variable edges in $G$, at least one of these six variable edges is claimed by Breaker, and so, Maker does not have six internally disjoint paths of length~$4$ from the clause vertex~$\Tilde{C}$ to six different 5-cycles, a contradiction.

Now, we prove that if Falsifier wins in $\phi$, then Maker wins the $H$-game in $G$. Assume that Falsifier wins in $\phi$. Maker first claims a variable edge in $G$ that corresponds to the variable she wants to set to false in $\phi$ according to her winning strategy as Falsifier in $\phi$. Maker then partitions the edge set of the base graph into several parts, and plays on each of them separately, responding in each of them as the second player. She plays arbitrarily in $G$ whenever she cannot follow her strategy, and then she resumes her strategy. Specifically, she uses Falsifier's winning strategy in $\phi$ on the variable edges of $G$. On each copy of $Y$, she claims a $C_5$ through the attachment vertex, as guaranteed by Claim~\ref{c:Y}. In each copy of $Z$, she claims a path connecting the end vertices, as guaranteed by Claim~\ref{c:Z}. Finally, in each copy of $K_6$, she claims a $C_3$ on the attachment vertex, as guaranteed by Claim~\ref{c:K6}.

Now, we can spot a clause for which all six of the variable edges are claimed by Maker, and locate a Maker's copy of $H$ on that clause's vertex. Indeed, Maker's graph will contain a triangle on that clause's vertex, six paths of length~$3$ that connect that clause's vertex with the six Maker's variable edges, and on the ``upper'' vertex of each of those edges, there will be a Maker's 5-cycle. 
\end{proof}

\section{Linear-time Algorithms}\label{sec:poly}

In this section, we give linear-time algorithms to solve the $P_4$-game in any graph, and the $K_{1,\ell}$-game in trees for any fixed integer $\ell\geq 1$. 

\subsection{Linear-time algorithm for the $P_4$-game}

The $P_4$ game can be solved in polynomial time by the algorithm of Galliot {\it et al.} for Maker-Breaker games with winning sets of size at most three~\cite{Galliot22}. However, this algorithm does not give any structural insight about the graphs in which Maker (Breaker, resp.) wins, and it runs in time $O(|V(\mathcal{H})|^5|E(\mathcal{H})|^2+|V(\mathcal{H})|^6\Delta(\mathcal{H}))$, where $\mathcal{H}$ is the hypergraph for the game. We give a necessary and sufficient structural condition for Breaker to win the $P_4$-game in any graph $G$. This leads to a linear-time algorithm to decide the outcome of the $P_4$-game.  

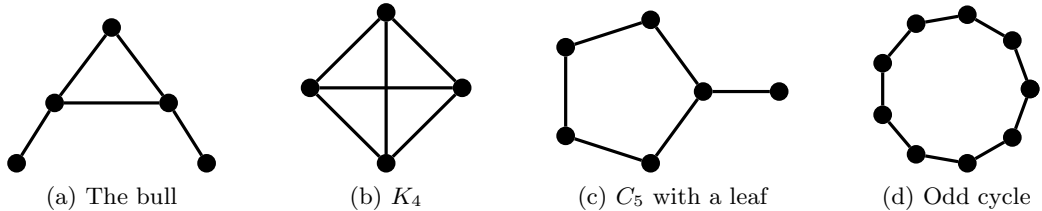
\begin{figure}
    \centering
    \subfloat[The bull]{\label{bull}
        \begin{tikzpicture}
        \node[v](0) at (0,0) {};
          \node[v](1) at (1.5,0) {};
            \node[v](2) at (0.75,1) {};
          \node[v](3) at (-0.5,-0.8) {};
             \node[v](4) at (2,-0.8) {};
          
          \draw[e] (0)--(1)--(2)--(0)--(3) (1)--(4);
        \end{tikzpicture}
 }
     \hfil
   \subfloat[$K_4$]{\label{diamond}
          \begin{tikzpicture}
        \node[v](0) at (0,0) {};
          \node[v](1) at (2,0) {};
            \node[v](2) at (1,-1) {};
          \node[v](3) at (1,1) {};
          \draw[e] (0)--(1)--(2)--(0)--(3)--(1);
          \draw[e] (3)--(2);
        \end{tikzpicture}
       }
     \hfil
    \subfloat[$C_5$ with a leaf]{\label{oddpath}
                \begin{tikzpicture}
        \foreach \I in {0,1,...,4}
    { \node[v](\I) at (\I*72:1) {};}
  
   \node[v](00) at (2,0) {};
    
          \draw[e] (0)--(1)--(2)--(3)--(4)--(0);
           \draw[e] (0)--(00);
        \end{tikzpicture}
       }
         \hfil
    \subfloat[Odd cycle]{\label{oddcycle}
                \begin{tikzpicture}
        \foreach \I in {0,1,...,9}
    { \node[v](\I) at (\I*40:1) {};}

          \draw[e] (0)--(1)--(2)--(3)--(4)--(5)--(6)--(7)--(8)--(0);
        
        \end{tikzpicture}
       }
    \caption{Maximal non-bipartite graphs for which Breaker wins the $P_4$-game.}
    \label{fig:exceptions}
\end{figure}

\begin{theorem}\label{thm:P4game}
For any connected graph $G$, Breaker wins the $P_4$-game in $G$ if and only if

\begin{enumerate}
\item $G$ is bipartite and all the vertices of degree at least 3 are in the same part; or

\item $G$ is an odd cycle; or

\item $G$ is a subgraph of the bull, $K_4$, or a $C_5$ with a leaf attached to one vertex (see Figure~\ref{fig:exceptions}).
\end{enumerate}
\end{theorem}

\begin{proof}
We first prove the ``if'' part.
One can check with a small case analysis that Breaker wins in the first three graphs of Figure~\ref{fig:exceptions}, and thus, in all their subgraphs by Lemma~\ref{subgraph}.
For the odd cycle, Breaker wins with the following strategy. Let $(e_1,\dots,e_{2\ell+1})$ be the edges of the cycle in this order. Without loss of generality, let $e_1$ be the first edge claimed by Maker. Then, Breaker claims $e_{2\ell+1}$, and then follows a pairing strategy by pairing edges $e_{2i}$ with $e_{2i+1}$, for $1\leq i<\ell$ ($e_{2\ell}$ is not paired). 

Thus, we can assume that $G$ is bipartite with all the vertices of degree at least $3$ in the same part. Let $V(G)=A\cup B$ be a bipartition of the vertices of $G$ with the part $A$ containing all the vertices of degree at least 3, and so, the vertices of $B$ have degree at most 2.
Note that any path on four vertices must contain an inner vertex in $B$, and thus, two edges incident to the same vertex in $B$. Thus, Breaker can win by following a pairing strategy where, for each vertex $v$ of $B$, the (potentially) two edges incident to $v$ are paired together.

We now prove the ``only if'' part. Let $G$ be a graph that does not satisfy the three conditions of Theorem \ref{thm:P4game}. The proof is divided into four cases:

\begin{enumerate}
\item $G$ contains a diamond ($K_4$ minus one edge).

\item $G$ contains a triangle, but not a diamond.

\item $G$ contains an odd path between two vertices of degree at least 3, but not a triangle.

\item $G$ contains an odd cycle with a unique vertex of degree at least 3.  
\end{enumerate}

These cases will cover all the possible graphs where Maker wins. Indeed, if $G$ is bipartite, then it would be treated in Case 3. If $G$ is not bipartite and not an odd cycle, then it has a vertex of degree at least $3$ in an odd cycle $C$. If $C$ contains a unique vertex of degree at least~$3$, then we are in Case~4. Otherwise, $C$ contains two vertices of degree at least 3, and thus, there is an odd path between vertices of degree at least 3, and we are in Case 3 if $G$ is triangle-free. Otherwise, $G$ contains a triangle, but not a diamond (Case~2) or $G$ contains a diamond (Case~1).

\paragraph{Case 1. $G$ contains a diamond.} Since $G$ is not restricted to a subgraph of $K_4$, $G$ contains a subgraph that is a diamond with a leaf connected to it (either to a vertex of degree 3 or to a vertex of degree 2 in the diamond). One can easily check that Maker wins the $P_4$-game in these subgraphs, and thus, by Lemma~\ref{subgraph}, she wins in all the graphs containing them.

\paragraph{Case 2. $G$ contains a triangle, but not a diamond.} Let $uvw$ be this triangle. Since $G$ is not a triangle (it would be a subgraph of the bull), at least one vertex, say $u$, must be connected to another vertex $z$. Note that $z$ is not connected to any other vertex of the triangle, as otherwise $G$ would contain a diamond. Since $G$ is not restricted to the graph on the vertices $\{u,v,w,z\}$ (because otherwise it would be a subgraph of $K_4$), there is another vertex $x$ in the graph.

Assume first that $x$ is connected to $z$. Then, Maker claims $zu$. Then, she has a pairing strategy with pairs $(xz,vw)$ and $(uv,uw)$. Assume now that $x$ is connected to $u$. Then, Maker claims $vw$, and then follows a pairing strategy with pairs $(xu,zu)$ and $(uv,uw)$.
Finally, assume that $x$ is connected to the triangle, without loss of generality, by $v$. Since $G$ is not a bull and does not contain a diamond, there must either be an additional edge, and it can only be $xz$ (otherwise $G$ would contain a diamond) or there is another vertex connected to the graph. The first case returns to the case where $x$ was connected to $z$. For the second case, the only possibility not covered yet is that there is a vertex $t$ connected to $w$. Then, Maker claims $uw$. If Breaker then claims $vx$ or $vu$, then Maker claims $uz$, and pairs $tw$ with $wv$.
Otherwise, without loss of generality, Breaker then claims $uz$. Then, Maker claims $uv$ and pairs $tw$ with $vx$.

From now on, we can assume that $G$ is triangle-free. Thus, Maker just needs to claim any three consecutively incident edges and will be sure to obtain a $P_4$.

\paragraph{Case 3. $G$ contains an odd path between two vertices of degree at least 3, but not a triangle.} Let $u$ and $v$  be two vertices of degree at least 3 connected by an odd path $P$. We choose $u$ and $v$ such that $P$ has minimum length. Let $e_1,\dots,e_{2\ell+1}$ be the edges of $P$ with $e_1$ incident to $u$, and $e_{2\ell+1}$ incident to $v$. Let $e_u$ and $e'_u$ be the two other edges incident to $u$, and $e_v$ and $e'_v$ the two other edges incident to $v$.
 
If $\ell=0$, then Maker claims $e_1$ and follows a pairing strategy with pairs $(e_u,e'_u)$ and $(e_v,e'_v)$ ($e_u$ and $e'_u$ are vertex-disjoint from $e_v$ and $e'_v$ since $G$ is triangle-free). Assume now that $\ell>0$. Maker starts by claiming $e_2$. Assume first Breaker does not answer by claiming $e_1$. Then, Maker can claim $e_1$ as her second move. Then, either she pairs $e_u$ with $e'_u$ if Breaker did not claim any of these edges on his first move, or she pairs $e_u$ or $e'_u$ (the one that is unclaimed) with $e_3$. Therefore, we can assume that Breaker answers by claiming $e_1$.

Assume now, by induction on $1<i<\ell$, that before their $i^{th}$ moves, Maker has claimed all the even edges $e_{2j}$, and Breaker has claimed all the odd edges $e_{2j-1}$ for $j=1,\dots,i-1$ (we have shown this is true for $i=2$ above). Then, on her $i^{th}$ move, Maker claims $e_{2i}$. Breaker has to answer by claiming $e_{2i-1}$, since otherwise Maker can claim $e_{2i-1}$, creating a $P_4$ with $e_{2i}$ and $e_{2i-2}$. Then, the inductive hypothesis holds for $i+1$.
Repeating this argument, for her $\ell^{th}$ move, Maker claims $e_{2\ell}$, and Breaker has to answer by claiming $e_{2\ell-1}$. Then, Maker claims $e_{2\ell+1}$ and pairs $e_v$ with $e'_v$, which will create a $P_4$.

\paragraph{Case 4. $G$ contains an odd cycle with a unique vertex of degree at least 3.}
Let $u$ be the unique vertex of degree at least $3$ in an odd cycle in $G$. Let $e_1,\dots,e_{2\ell+1}$ be the edges of the cycle, with $u$ incident to $e_1$ and $e_{2\ell+1}$. Let $v$ be a vertex adjacent to $u$, but not in the cycle (it exists since $u$ has degree at least $3$).
Since $G$ is triangle-free, $\ell>1$.
Assume first that $\ell=2$. Since $G$ is not restricted to a $C_5$ with a leaf, and since any additional edge will create another vertex of degree at least 3 in the cycle, there must be another vertex $w$ in the graph.
If $w$ is adjacent to $u$, then Maker claims $e_2$. Then, as before for the odd path, Breaker should answer by claiming $e_1$. Then, Maker claims $e_4$, and Breaker should answer by claiming $e_3$. Finally, Maker claims $e_5$, and then pairs $uv$ with $uw$.
If $w$ is adjacent to $v$, then Maker claims $uv$. Breaker should answer by claiming $vw$, as otherwise Maker can make a $P_3$ with two free extremities. Then, Maker can force moves by claiming $e_2$ (Breaker then claims $e_1$), then $e_4$ (Breaker then claims $e_3$), and then win by claiming $e_5$.

Assume now that $\ell>2$, {\it i.e.}, the cycle has length at least 7. Maker claims $e_4$. Breaker should claim either $e_3$ or $e_5$ to avoid a $P_3$ with two free extremities. If, without loss of generality, Breaker claims $e_3$, Maker can again force moves by claiming the even edges $e_6,\dots,e_{2\ell}$. Breaker always has to answer by claiming the preceding odd edges $e_5,\dots,e_{2\ell-1}$. Then, Maker claims $e_{2\ell+1}$ and pairs $e_1$ with $uv$, making a $P_4$.
\end{proof}

The conditions given in Theorem \ref{thm:P4game} are checkable in linear time, which implies the following:

\begin{corollary}
It can be decided in linear time whether Maker wins the $P_4$-game in a given connected graph $G$.
\end{corollary}

Solving the $P_{\ell}$-game for $\ell>4$ seems difficult, even for trees, since Maker's winning strategy can be highly non-trivial. Indeed, already for $\ell=4$, there are trees for which Maker needs an unbounded number of moves to win and needs to play in a disconnected way. An example is given in Figure~\ref{fig:badtree}, where there can be an arbitrary even number of vertices of degree 2 in the middle path. Moreover, this can be generalized for the $P_{\ell}$-game for any $\ell\geq 4$.

\begin{figure}
    \centering
    \begin{tikzpicture}
   \node[v](a) at (-0.5,-0.5) {};
   \node[v](b) at (-0.5,0.5) {};
   \foreach \I in {0,...,7}
   {\node[v](\I) at (\I,0) {};}
     \node[v](c) at (7.5,-0.5) {};
   \node[v](d) at (7.5,0.5) {};
   \draw (a)--(0)--(1)--(2)--(3)--(4)--(5)--(6)--(7)--(c);
   \draw (b)--(0) (d)--(7);
   
    \end{tikzpicture}
    \caption{A tree where Maker wins the $P_4$-game, but must play disconnected. By adding an even number of degree-$2$ vertices in the middle path, this gives a family of trees where the number of moves to win for Maker is unbounded.}
    \label{fig:badtree}
\end{figure}
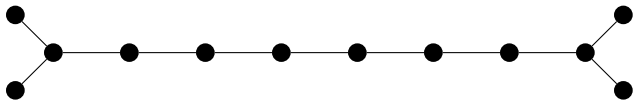

However, if one extremity of the path is fixed, the game becomes easy.
Consider the rooted path game in a tree $T$ rooted in $r$: Maker wants to make the longest possible path that starts from $r$. Let $\ell_M(T,r)$ ($\ell_B(T,r)$, respectively) be the longest path starting from $r$ that Maker can construct when she plays first (when Breaker plays first, resp.).

\begin{lemma}
Let $T$ be a tree rooted in $r$, $x_1,\dots,x_s$ the children of $r$, and $T_1,\dots,T_s$ the subtrees rooted in $x_1,\dots,x_s$. Without loss of generality, assume that $\ell_B(T_1,x_1)\geq \ell_B(T_2,x_2) \cdots \geq \ell_B(T_s,x_s)$.
Then, $\ell_M(T,r)=\ell_B(T,r)=0$ if $r$ is the single vertex in $T$. Otherwise:
$$\ell_M(T,r)=1+\ell_B(T_1,x_1)$$
$$\ell_B(T,r)=1+\ell_B(T_2,x_2).$$

In particular, they can be computed in linear time.
\end{lemma}

\begin{proof}
Consider first $\ell_M(T,r)$. Maker can construct a path of length $1+\ell_B(T_1,x_1)$ by claiming the edge $rx_1$, and then following her strategy as the second player in $T_1$ rooted in $x_1$.
If Maker claims $rx_i$ as her first edge, then Breaker can assure her path has length at most $1+\ell_B(T_i,x_i)\leq 1+\ell_B(T_1,x_1)$ by following his strategy in $T_i$ rooted in $x_i$.

Consider now Breaker as the first player. By claiming $rx_1$, the game is then equivalent to the game in the tree $T'$ where $T_1$ has been removed and Maker starts. As before, Maker will claim a path of length $1+\ell_B(T_2,x_2)$.
Breaker cannot do better: if he claims another edge that is not in $T_1$, then Maker claims $rx_1$ and can construct a path of length $1+\ell_B(T_1,x_1)$, which is worse for Breaker. If he claims an edge $e=uv$ in $T_1$, the game is equivalent to the game played in $T''$, the tree obtained by cutting the tree on the edge $e$ and keeping the part containing $r$. Since $T''$ contains $T'$, Maker will be able to construct a path of length at least $1+\ell_B(T_2,x_2)$.
\end{proof}

\subsection{$K_{1,\ell}$-game in trees}

In this section, we consider the $K_{1,\ell}$-game in trees. Otherwise said, Maker needs to claim $\ell$ edges adjacent to the same vertex. We prove that this game is tractable in trees.

\begin{theorem}\label{thm:startree}
For any tree $T$ and any fixed integer $\ell \geq 1$, it can be decided in linear time whether Maker wins the $K_{1,\ell}$-game in $T$.
\end{theorem}

To prove this theorem, we need three structural lemmas. The first one is true for the $K_{1,\ell}$-game in any graph.

\begin{lemma}\label{lem:highdegree}
For any graph $G$ and any fixed integer $\ell \geq 1$, if $G$ contains a vertex of degree at least $2\ell-1$, then Maker wins the $K_{1,\ell}$-game in $G$.
\end{lemma}

\begin{proof}
Let $u$ be a vertex of degree at least $2\ell-1$. Maker claims any $\ell$ edges incident to $u$ in her first $\ell$ moves.
\end{proof}

\begin{lemma}\label{lem:lowdegree}    
For any fixed integer $\ell \geq 1$, if $T$ is a tree with maximum degree at most $2\ell-2$ and at most one vertex of degree $2\ell-2$, then Breaker wins the $K_{1,\ell}$-game in $T$.
\end{lemma}

\begin{proof}
Let $r$ be a vertex of maximum degree (possibly $2\ell-2$) and root $T$ in $r$.
We define a pairing strategy for Breaker as follows.
For any vertex $u\in V(T)$, let $v_1,\dots,v_{t_u}$ be its children, and pair together the edges $\{uv_{2i-1},uv_{2i}\}$ for $i=1$ to $i=\lfloor t_u/2 \rfloor$. Breaker, playing this pairing strategy, will claim at least half of the edges incident to $r$ (so Maker will claim at most $\ell-1$ edges incident to $r$) and, for any other vertex $w\in V(T)$, Breaker will claim at least $\lfloor \frac{t_w-1}{2}\rfloor$ edges incident to $w$. Since $t_w<2\ell-2$ as $r$ is the only vertex that could have degree $2\ell-2$, Maker will claim at most $\ell-1$ edges incident to $u$.
\end{proof}

The next lemma will enable us to cut a tree into several components when there is a vertex of degree $2\ell-2$. We first describe the cut operation we are using. Let $T$ be a tree and let $uw\in E(T)$. Let $T_u$ and $T_w$ be the two trees composing the forest $T\setminus \{uw\}$, with $u\in V(T_u)$ and $w\in V(T_w)$. Let $T_1$ be the tree obtained from $T_u$ by adding one pendent edge incident to $u$, and let $T_2$ be the tree obtained from $T_w$ by adding two pendent edges incident to $w$. We call the forest $T_1\cup T_2$ a $(T,u,uw)$-{\em cut}. See Figure \ref{fig:cut} for an illustration.

\begin{figure}
    \centering
    \begin{tikzpicture}
    \node[v](u) at (0,0) {};
    \node[v](w) at (1,0) {};
    \node[below=5pt] at (u) {$u$};
    \node[below=5pt] at (w) {$w$};
    \draw[e] (u)--(w);
    \draw[draw] (u) -- (-1.2,0.8) -- (-1.2,-0.8) -- (u);
    \node at (-0.8,0) {$T_u$};
     \draw[draw] (w) -- (2.2,0.8) -- (2.2,-0.8) -- (w);
    \node at (1.8,0) {$T_w$};
    \node at (0.5,-1) {$T$};

    \draw[->] (3.5,0) to node[above] {$(T,u,uw)$-cut} (5,0);
    
    \begin{scope}[shift={(8,0)}]
        \node[v](u) at (0,0) {};
    \node[v](w1) at (1,0) {};
      \node[v](u1) at (2,0.5) {};
    \node[v](u2) at (2,-0.5) {};
     \node[v](w) at (3,0) {};
    \node[below=5pt] at (u) {$u$};
    \node[below=5pt] at (w) {$w$};

\draw[e] (u)--(w1);
\draw[e] (u1)--(w)--(u2);
    \draw[draw] (u) -- (-1.2,0.8) -- (-1.2,-0.8) -- (u);
    \node at (-0.6,0) {$T_u$};
     \draw[draw] (w) -- (4.2,0.8) -- (4.2,-0.8) -- (w);
    \node at (3.6,0) {$T_w$};

    \node at (0.2,-1) {$T_1$};
    \node at (3.2,-1) {$T_2$};
    \end{scope}
    \end{tikzpicture}
    \caption{Illustration of a $(T,u,uw)$-cut. If $u$ has degree $2\ell-2$, then the $K_{1,\ell}$-game in $T$ or in the cut has the same outcome.}
    \label{fig:cut}
\end{figure}
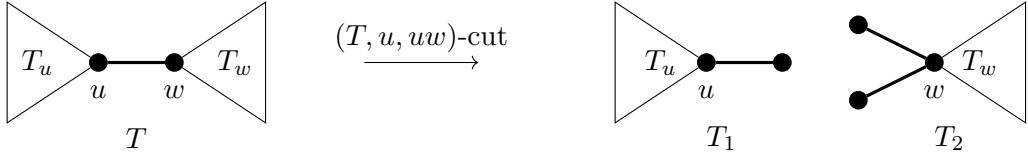

\begin{lemma}\label{lem:decomposition}
For any fixed integer $\ell \geq 1$, let $T$ be a tree with a vertex $u$ of degree $2\ell-2$, $w$ any neighbor of $u$, and $T_1\cup T_2$ the $(T,u,uw)$-cut. Maker wins the $K_{1,\ell}$-game in $T$ if and only if she wins the $K_{1,\ell}$-game in $T_1\cup T_2$.
\end{lemma}

\begin{proof}
Let $e$ be the edge added to $u$ in $T_1$, and $e_1,e_2$ the two edges added to $w$ in $T_2$.
Assume that Maker wins in $T$. In $T_1\cup T_2$, Maker plays as follows. She pairs the edges  $e_1$ and $e_2$, and associates all the other edges with their corresponding edges in $T$ ($e$ is associated with $uw$). Then, she plays as in $T$. If Breaker claims $e_1$ or $e_2$, then she claims the other edge. At the end of the game in $T$, there is a vertex $x$ in $T$ that has $\ell$ incident edges claimed by Maker. If $x\neq w$, it has the same $\ell$ incident edges at the end of the game in $T_1\cup T_2$ (with possibly the edge $uw$ replaced by $e$ if $x=u$). If $x=w$, and Maker claimed the edge $uw$ that is not present in $T_2$, in the game in $T$, then Maker claimed the same $\ell$ edges in the game in $T_1\cup T_2$ except that $uw$ has been replaced by one of the edges $e_1$ or $e_2$.

For the other direction, assume that Maker has a winning strategy in $T_1\cup T_2$. By Lemma~\ref{lem:union}, she wins either in $T_1$ or in $T_2$. 
Assume she wins in $T_1$. Then, she can follow the same strategy in $T$ without taking care of Breaker's moves in the rest of $T$ (and considering that $e=uw$).
Assume now she wins in $T_2$. By Lemma \ref{super lemma}, since $e_1$ and $e_2$ are interchangeable in any winning set that contains one of them, Maker also wins in $(T_2,\{e_1\},\{e_2\})$, ({\it i.e.}, $T_2$ where $e_1$ is claimed by Maker and $e_2$ by Breaker).

To win in $T$, Maker first claims $uw$. Then, Breaker should answer by claiming an edge incident to $u$ since $u$ has degree $2\ell-2$ (otherwise, Maker can claim $\ell$ edges incident to $u$). Then, Maker follows her winning strategy in $(T_2,\{e_1\},\{e_2\})$ in $T_w$. The unclaimed edges are in a one-to-one correspondence, and the vertices have the same number of edges claimed by Maker (that is, one for $w$, and $0$ for the other vertices). Thus, Maker will win in $T$.
\end{proof}

The next theorem gives a necessary and sufficient structural condition for Maker to win the $K_{1,\ell}$-game in trees. This will imply Theorem~\ref{thm:startree} since it is easy to check if a tree has this structure (see the proof of Theorem~\ref{thm:startree} below).

\begin{theorem}\label{structurestar}
Let $T$ be a tree and $\ell\geq 1$ a fixed integer.
Maker wins the $K_{1,\ell}$-game in $T$ if and only if there is a subtree $T'$ of $T$ such that every vertex $x$ of $T'$ has degree at least $2\ell-1-d_{T'}(x)$ in $T$, where $d_{T'}(x)$ denotes the degree of $x$ in $T'$. 
\end{theorem}

Note that if $T'$ is reduced to a single vertex, it must have degree $2\ell-1$ in $T$, which corresponds to Lemma \ref{lem:highdegree}. Otherwise, all the leaves of $T'$ must be of degree $2\ell-2$ in $T$. The theorem says that these leaves must be connected by vertices of sufficiently large degree.

\begin{proof}
We prove the equivalence by induction on the number of inner vertices of $T$ ({\it i.e.}, vertices of degree at least 2).
Assume that $T$ has only one inner vertex $u$. The only possibility for $T'$ is to be the single vertex $u$, and Maker wins if and only if $u$ has degree $2\ell-1$. Thus, the equivalence is true.

Assume now, by induction, that the result is true for any tree with at most $i-1\geq 1$ inner vertices. Let $T$ be a tree with $i$ inner vertices. 
If $T$ has at least one vertex $u$ of degree at least $2\ell-1$, then take $T'$ to be the single vertex $u$, and Maker wins by Lemma~\ref{lem:highdegree}.
If $T$ has maximum degree at most $2\ell-2$, and $T$ contains at most one vertex of degree $2\ell-2$, it is not possible to have a subtree $T'$ satisfying the property. Indeed, since $T'$ cannot be a single vertex, it must have at least two leaves of degree $2\ell-2$ in $T$. By Lemma~\ref{lem:lowdegree}, Breaker wins, and so, the equivalence is true.

Thus, we can assume that $T$ has at least two vertices of degree exactly $2\ell-2$.
Let $u$ and $v$ be two vertices of degree $2\ell-2$ at maximum distance from each other. Let $w$ be the neighbor of $u$ on the path between $u$ and $v$. Let $T_1\cup T_2$ be the $(T,u,uw)$-cut.
Note that, by the maximality of the distance between $u$ and $v$, the vertex $u$ is the only vertex of $T_1$ of degree $2\ell-2$. In particular, Breaker wins in $T_1$ by Lemma~\ref{lem:lowdegree}.
By Lemma~\ref{lem:decomposition}, Maker wins in $T$ if and only if she wins in $T_1\cup T_2$. Then, since Breaker wins in $T_1$, by Lemma~\ref{lem:union}, Maker wins in $T$ if and only if she wins in $T_2$.
The tree $T_2$ has strictly less than $i$ inner vertices, and thus, by induction, Maker wins in $T_2$ if and only if there is a tree $T'_2$ such that every vertex of degree $t$ in $T'_2$ has degree
$2\ell-1-t$ in $T_2$.

Assume Maker wins in $T$. This means that such a $T'_2$ exists. If $w \notin V(T'_2)$, then it is a valid tree $T'$ for $T$. If $w \in V(T'_2)$, then $T'_2$ does not contain the two pendent edges incident to $w$ (since no leaf of $T_2$ can belong to $T'_2$). Then, let $T'$ be the subtree of $T$ obtained from $T'_2$ by adding $u$. Note that $T'_2$ cannot have leaves of $T$ as vertices, and thus, all its vertices are inner vertices of $T$. Since $u$ is a leaf of $T'$ and has degree $2\ell-2$ in $T$, it satisfies the degree condition. This is also true for all the other vertices of $T'$ except $w$. Let $t_2$ be the degree of $w$ in $T'_2$. By hypothesis, its degree in $T_2$, $d_{T_2}(w)$, is at least $2\ell-1-t_2$. The vertex $w$ has degree $t_2+1$ in $T'$ and $d_{T_2}(w)-1$ in $T$. Thus, $d_T(w)=d_{T_2}(w)-1\geq 2\ell-1-t_2-1=2\ell-1-d_{T'}(w)$, and $T'$ satisfies the property.

For the other direction, assume that there is a subtree $T'$ valid in $T$. Let $uw'$ be the pendent edge in $T_1$ ({\it i.e.}, $w'\notin V(T)$). If $T'$ is a subtree of $T_1 \setminus \{w'\}$, then, by induction, Maker wins in $T_1$, and, by Lemmas~\ref{lem:union}~and~\ref{lem:decomposition}, she wins in $T$. Otherwise, let $V(T'_2)=V(T')\cap V(T_2)$. As before, one can prove that the degree condition is correct for each vertex in $T'_2$ in $T_2$ (indeed, if $u$ was in $T'$, then $w$ has one less neighbor in $T'_2$ than it did in $T'$, but this is compensated by the fact that $w$ has one more neighbor in $T_2$). Thus, by induction, Maker wins in $T_2$, and so, by Lemmas~\ref{lem:union}~and~\ref{lem:decomposition}, Maker wins in $T$.

Therefore, the equivalence is true for $T$, and, by induction, is true for all $T$.
\end{proof}

This characterization implies a linear-time algorithm for this game.

\begin{proof}[Proof of Theorem \ref{thm:startree}]
To find such a $T'$, one can do a breadth-first search starting from any vertex of $T$. Label each vertex with its degree in $T$. Then, consider each vertex $x$, starting from the deepest level. If $x$ has label $2\ell-2$, then increment the label of the parent of $x$ by 1, and otherwise, do nothing.
If at some point a vertex is labeled with $2\ell-1$, then Maker wins, and otherwise, Breaker wins. Note that when we increment the label of a vertex, it corresponds to the cut operation. Precisely, if a vertex $v$ receives a label $2\ell-1$, let $T_v$ be the subtree rooted in $v$. Then, a subtree $T'$ satisfying the requirement is the inclusion-minimal subtree of $T_v$ containing $v$ and whose leaves ($\neq v$) have degree $2\ell-2$ in $T$.
\end{proof}

Whether the $K_{1,\ell}$-game is polynomial-time solvable for general graphs is still open, but we prove in the next section that it is \FPT\ parameterized by the length of the game.

\section{\FPT\ Algorithms}\label{sec:FPT}

In this section, we consider the $H$-game parameterized by the length of the game $k$. That is, given a fixed integer $k \in \mathbb{N}$, a fixed graph $H$, and a graph $G$ as input, the problem consists in deciding whether Maker wins the $H$-game in at most $k$ moves. 


In particular, we prove that determining the outcome for the $K_{1,\ell}$-game in any graph and the $H$-game in any tree are both \FPT\ parameterized by $k$.\footnote{Note that it may be possible to obtain these results via meta-theorems~\cite{Bonnet2020,Courcelle1990, Frick2001} combined with adapting the arguments of Bonnet et al.~\cite{Bonnet2017} to express the winning conditions in first-order logic formulas with $k$~quantifiers. However, compared to our results, these meta-theorems would not provide an optimal strategy for the winning player and they can entail much larger running times.} The next theorem is crucial for these results, and is interesting on its own since it allows to bound the diameter of the graph after Maker's first move, which could lead to other positive results for the $H$-game, and the proof technique could be generalized to other games on graphs.

For any graph $G$, $e=uv \in E(G)$, and $r \in \mathbb{N}$, let $B_G(e,r)=\{ww' \in E(G) \mid \max \{dist_G(u,w),$ $dist_G(v,w),dist_G(u,w'),dist_G(v,w')\} \leq r\}$. Given a graph $G$ and $X,Y \subseteq E(G)$ (with $|X|=|Y|+1$), Breaker wins the $H$-game in position $(G,X,Y)$ in $i\geq 1$ moves if it is Breaker's turn and Maker cannot create $H$ in at most $i-1$ moves after Breaker's next move.

\begin{theorem}\label{thm:balls}
Let $H$ be a connected graph, $G$ any graph, and $k$ a positive integer. Maker wins the $H$-game in $G$ in at most $k$ moves if and only if there exists $e \in E(G)$ such that Maker wins the $H$-game in $G[B_G(e,3^k)]$ in at most $k$ moves.
\end{theorem}

\begin{proof}
First, if there exists $e \in E(G)$ such that Maker wins the $H$-game in $G[B_G(e,3^k)]$ in at most $k$ moves, then Maker wins the $H$-game in $G$ in at most $k$ moves by Lemma~\ref{subgraph}. Now, assume that Breaker wins the $H$-game in $G[B_G(e,3^k)]$ in at most $k$ moves for each $e \in E(G)$. We describe a winning strategy for Breaker in the $H$-game in $G$ that takes at most $k$ moves.

Let $e_1$ be the first edge claimed by Maker, and let $G_1=G[B_G(e_1,3^k)]$. Note that  Breaker wins the $H$ game in $(G_1,\{e_1\},\emptyset)$ in at most $k$ moves by the initial assumption. In particular, if Maker always claims an edge in $G_1$, then Breaker wins. 

First, Breaker answers to Maker claiming $e_1$ by following his winning strategy in $G_1$. Let $2 \leq i \leq k$ be the $i^{th}$ round of the game, before the $i^{th}$ move of Maker, and let $M_i=\{e_1,\dots,e_{i-1}\}$ be the edges claimed by Maker, and $B_i=\{f_1,\dots,f_{i-1}\}$ the edges claimed by Breaker. Assume, by induction on $i$, that there exist edge-disjoint subgraphs $G_1,\dots,G_{s_i}$ such that:

\begin{itemize}
\item for every $1 \leq j \leq s_i$, there exists $e_j \in E(G)$ such that $E(G_j) \subseteq B_G(e_j,3^k)$;

\item for every $e \in M_i$, there exists a unique $1\leq j^i_e \leq s_i$ such that $e \in E(G_{j^i_e})$ and, moreover, $B_G(e,3^{k-i+1}) \subseteq E(G_{j^i_e})$;

\item for every $1 \leq j \leq s_i$, Breaker wins the $H$-game in $(G_j,E(G_j) \cap M_i,E(G_j) \cap B_i)$ in at most $k-i+1$ moves.
\end{itemize}

The inductive hypothesis holds if $i=2$ by remarks above (in particular, $s_2=1$). Assume that the inductive hypothesis holds for $i\geq 2$. Let $e_i$ be the $i^{th}$ edge claimed by Maker.

\begin{itemize}
\item If there exists $1 \leq j \leq s_i$ such that $B_G(e_i,3^{k-i}) \subseteq E(G_{j})$, then Breaker answers by following his winning strategy in $G_j$. Note that, in this case, $j$ is unique since $G_1,\dots,G_{s_i}$ are edge-disjoint subgraphs by the inductive hypothesis for $i$. Then, the inductive hypothesis holds for $i+1$ with the same subgraphs $G_1,\dots,G_{s_i}$ (in particular, $s_i=s_{i+1}$).

\item If $B_G(e_i,3^{k-i}) \cap E(G_j) = \emptyset$ for all $1 \leq j \leq s_i$, then let $s_{i+1}=s_i+1$ and $G_{s_{i+1}}=G[B_G(e_i,3^{k-i})]$. Then, Breaker answers by following his winning strategy in $G_{s_{i+1}}$, which exists by the assumption that Breaker wins the $H$-game in $G[B_G(e,3^k)]$ in at most $k$ moves for each $e\in E(G)$, and since $G_1,\dots,G_{s_{i+1}}$ are edge-disjoint subgraphs by the inductive hypothesis for $i$ and the case we are in. Indeed, by Lemma~\ref{subgraph}, Breaker has a winning strategy in $G[B_G(e_i,3^{k-i})]$ in $k-i$ moves since he has one in $G[B_G(e_i,3^k)]$ in $k$ moves (if Maker cannot create $H$ in $k$ moves, then she clearly cannot do it in $k-i$ moves in a subgraph). Then, the inductive hypothesis holds for $i+1$.

\item Lastly, if there exists $\emptyset \neq J \subseteq \{1,\dots,s_i\}$ such that, for all $j \in J$, $B_G(e_i,3^{k-i}) \cap E(G_j) \neq  \emptyset$ and $B_G(e_i,3^{k-i}) \setminus E(G_j) \neq  \emptyset$, then let $s_{i+1}=s_i+1$ and $G_{s_{i+1}}=G[B_G(e_i,3^{k-i})]$. Now, for every $j \in J$, let $E_j=\{f \in E(G_j) \mid B(f,2\cdot3^{k-i}) \nsubseteq E(G_j)\}$. Note that, for every $f' \in M_i$ and $j\in J$, $E_j \cap B(f',3^{k-i}) = \emptyset$ by the second assumption of the inductive hypothesis for $i$. For all $j\in J$, let $G_j = G[E(G_j) \setminus E_j]$ (intuitively, the edges of $G_j$ that are ``too close'' to $G_{s_{i+1}}$ are removed from $G_j$), and note that $G_1,\dots,G_{s_{i+1}}$ are now edge-disjoint subgraphs since $G_1,\dots,G_{s_i}$ were edge-disjoint subgraphs by the inductive hypothesis for $i$. Now, Breaker plays his next move according to his winning strategy in $G_{s_{i+1}}$, which, as in the previous case, exists by the assumption that Breaker wins the $H$-game in $G[B_G(e,3^k)]$ in at most $k$ moves for each $e\in E(G)$, and since $G_1,\dots,G_{s_{i+1}}$ are edge-disjoint subgraphs. Then, the inductive hypothesis holds for $i+1$.
\end{itemize}

The inductive hypothesis and the strategy described above guarantee that Breaker wins, {\it i.e.}, Maker cannot win in $G$ in at most $k$ moves.
\end{proof}

With Theorem~\ref{thm:balls} in hand, we now have one of the main tools to prove our \FPT\ results, which rely on the fact that we only need to consider the ball (of edges) of bounded diameter in the length of the game centered at the first edge claimed by Maker.

\begin{corollary}\label{FPT-star-game}
For any graph $G$ and any fixed integer $\ell\geq 1$, deciding whether Maker wins the $K_{1,\ell}$-game in $G$ is \FPT\ parameterized by the length of the game.
\end{corollary}

\begin{proof}
Consider the $K_{1,\ell}$-game, for a positive constant $\ell$ (recall that $H$ is a fixed graph in the $H$-game), in any graph $G$. Let $k$ be the length of the game. If there is a vertex of degree at least $2\ell-1$, then Maker wins in $\ell$ moves by Lemma~\ref{lem:highdegree}. Hence, we can assume that the maximum degree is at most $2\ell-2$.
By Theorem \ref{thm:balls}, Maker wins in $G$ in $k$ moves if and only if she wins in $k$ moves in one of the balls $B(e,3^k)$ for some edge $e\in E(G)$.
Since $\Delta(G)\leq 2\ell-2$, for any edge $f\in E(G)$, the ball $B(f,3^k)$ has size at most $(2(\Delta(G)-1))^{3^k}=(4\ell-6)^{3^k}$, {\it i.e.}, a function of $k$ since $\ell$ is a constant. Therefore, one can check if Maker wins by first checking the maximum degree, and then checking the outcomes of all possible games in the $|E(G)|$ balls (of edges) of diameter $3^k$ which have size bounded by a function $f(k)$. Indeed, this leads to an \FPT\ algorithm since, in any graph of size bounded by a function $f(k)$, the output of the $H$-game in at most $k$ moves can be determined by an exhaustive search in time $f'(k)$ for some computable function $f'$ (the length of the game is $k$, and the number of possible moves at each step is at most the number of edges which is at most $f(k)$).
\end{proof}

Theorem~\ref{thm:balls} combined with the particular structure of trees leads to the following result.

\begin{theorem}
For any connected graph $H$ and any tree $T$, deciding whether Maker wins the $H$-game in $T$ is \FPT\ parameterized by the length of the game.
\end{theorem}

\begin{proof}
First, note that if $H$ is not a tree, then Breaker trivially wins the $H$-game in $T$, and so, we can assume that $H$ is a tree. We prove that the $H$-game parameterized by the length of the game $k$ admits a kernel in $T$. That is, from $T$, we build, in polynomial time, a forest $F$ of size at most a function of $k$ (precised below) such that Maker wins in $T$ in at most $k$ moves if and only if there exists a connected component of $F$ in which Maker wins in at most $k$ moves.

First, for every edge $e \in E(T)$, let $T_e$ be the subtree of $T$ induced by the edges at distance at most $3^k$ from $e$, {\it i.e.}, $T_e$ is the subtree induced by $B(e,3^k)$. Let $F$ be the forest that consists of the disjoint union of the $T_e$'s, $e \in E(T)$. By Theorem~\ref{thm:balls}, Maker wins in $T$ in at most $k$ moves if and only if there exists $e \in E(T)$ such that Maker wins in $T_e$ in at most $k$ moves.

The {\it depth} of a rooted tree is the maximum distance from its root to a leaf. For every $e \in E(T)$, let us root $T_e$ in such a way that it has depth $d_e \leq 3^k$ (this is possible by the definition of $T_e$). A vertex $v \in V(T_e)$ has {\it level} $i\geq 0$ if the subtree of $T_e$ rooted in $v$ has depth $i$. Let us iteratively, for $i=1$ to $d_e$, replace $T^{i-1}_e$ ($T_e=T^0_e$) by a tree $T^i_e$ such that: Maker wins in $T^{i-1}_e$ in at most $k$ moves if and only if Maker wins in at most $k$ moves in $T^i_e$; and, for every vertex $v \in V(T^i_e$) at level $i$, the subtree of $T^i_e$ rooted in $v$ has size at most $n_i(k)$ (a function of $k$ whose recursive definition is given below). 

First, for $i=1$, for every vertex $v$ at level $1$ in $T_e$, ({\it i.e.}, all children of $v$ are leaves), if $v$ has more than $2k$ children, then remove all but $2k$ of its children. Let $T^1_e$ be the obtained tree. By construction, every vertex $v \in V(T^1_e)$ at level $1$ is the root of a subtree of size at most $n_1(k)=2k+1$. Moreover, since, for every vertex $v$ at level $1$ in $T_e$, at most $2k$ edges between $v$ and leaves can be claimed (as the length of the game is $k$), then the output of the $H$-game is the same in $T_e$ and $T^1_e$. 

Now, by induction on $i\geq 1$, let us assume that we have built a tree $T^i_e$ such that Maker wins in $T^i_e$ in at most $k$ moves if and only if Maker wins in at most $k$ moves in $T_e$; and, for every vertex $v \in V(T^i_e$) at level $i$, the subtree of $T^i_e$ rooted in $v$ has size at most $n_i(k)$. Let $g_i(k)$ be the number of rooted trees of depth at most $i$ and of size at most $n_i(k)$. 
For every $v \in V(T^i_e)$ at level $i+1$, let $S_1,\dots,S_r$ be the subtrees rooted in the children of $v$ (note that each of these subtrees has depth at most $i$ and size at most $n_i(k)$). For every possible rooted subtree $S$ of depth $i$ and size at most $n_i(k)$, if there are more than $2k$ copies of $S$ in the multiset of trees $\{S_1,\dots,S_r\}$, then remove all but $2k$ copies of $S$. Let $T^{i+1}_e$ be the resulting tree (after having done the above process for every vertex at level $i+1$ of $T^i_e$). In $T^{i+1}_e$, every vertex at level $i+1$ has at most $2kg_i(k)$ children and those children are the roots of subtrees of size at most $n_i(k)$, and hence, every vertex at level $i+1$ is the root of a subtree of size at most $n_{i+1}(k)=2kg_i(k)n_i(k)$. Moreover, for every vertex $v$ at level $i+1$ in $T^i_e$, at most $2k$ edges in the subtree rooted at $v$ can be claimed. Therefore, the output of the $H$-game is the same in $T^i_e$ and $T^{i+1}_e$. 

After the above process has been done for $i=d_e \leq 3k$ for each subtree $T_e$, $e \in E(T)$, $F$ consists of the disjoint union of the trees $T^{3k}_e$, $e \in E(T)$, each of size at most $n_{3k}(k)$, and Maker wins in $T$ in at most $k$ moves if and only if she wins in at most $k$ moves in some connected component of $F$. While two of these subtrees are isomorphic, let us remove one of the two isomorphic subtrees. This clearly preserves the fact that Maker wins the $H$-game in $T$ in at most $k$ moves if and only if she wins in at most $k$ moves in some connected component of $F$. Moreover, eventually, $F$ has size at most $g_{3k}(k)n_{3k}(k)$, {\it i.e.}, this is the desired kernel. 

To conclude, this leads to an \FPT\ algorithm since, in any graph of size bounded by a function $f(k)$, the output of the $H$-game in at most $k$ moves can be determined by an exhaustive search in time $f'(k)$ for some computable function $f'$ (the length of the game is $k$, and the number of possible moves at each step is at most the number of edges which is at most $f(k)$).
\end{proof}

\section{Further Work}\label{sec:conclusion}

Since the $H$-game is \PSPACE-complete when $H$ is a tree, we investigated the complexity of the $H$-game when $H$ is a member of different subclasses of trees, leading to positive results. Another natural direction is to consider the case when $H$ is a cycle. While we were unable to resolve this case, we prove that the related arboricity-$k$ game is polynomial-time solvable. For any graph $G$, the arboricity of $G$, denoted by $\text{ar}(G)$, is the minimum number of forests into which its edges can be partitioned. In the arboricity-$k$ game played on the edge set of a graph, where $k\geq 2$, Maker wins if her graph at the end of the game has arboricity at least $k$. It is particularly interesting when $k=2$ since this is the cycle game (see, {\it e.g.},~\cite{bednarska2008odd}), the game in which Maker's goal is to claim a cycle, {\it i.e.}, $\mathcal{F}$ is the set of all the cycles in the graph.

\begin{theorem} \label{prop:arboricity}
    For any graph $G$, Maker wins the arboricity-$k$ game in $G$ if and only if $k\leq \left\lceil\text{ar}(G)/2 \right\rceil.$ Moreover, this can be decided in polynomial time.
\end{theorem}

\begin{proof}
Let $\ell:=\text{ar}(G)$. It was shown in the proof of~\cite[Proposition 17]{nenadov2016threshold} that Breaker can ensure that Maker's graph at the end of the game has arboricity at most $\left\lceil\ell/2 \right\rceil$. On the other hand, if we denote Maker's and Breaker's graphs at the end of the game by $M$ and $B$, respectively, then this defines a 2-partition of the edge set of $G$, and so, $\text{ar}(M)+\text{ar}(B)\geq \ell$.

We can apply strategy stealing to ensure the existence of a strategy for Maker that ensures her graph has high arboricity. Indeed, suppose that Breaker has a strategy to keep Maker's arboricity strictly below $\ell/2$. Then, by strategy stealing, Maker as the first player can apply the same strategy to keep Breaker's arboricity strictly below $\ell/2$. But, these two strategies put against each other would violate $\text{ar}(M)+\text{ar}(B)\geq \ell$, a contradiction. Hence, Maker can ensure that $\text{ar}(M)\geq \ell/2$ at the end of the game. Since $\text{ar}(M)$ is an integer, we are done.

Hence, to determine the outcome of the arboricity-$k$ game, it is enough to determine $\text{ar}(G)$, which can be done in polynomial time~\cite{picard1982network}.
\end{proof}

\begin{remark} In the proof of Theorem~\ref{prop:arboricity}, we use strategy stealing to prove the existence of a strategy for Maker, but this does not provide us with an explicit strategy. However, we can use the Nash-Williams arboricity theorem which gives $\text{ar}(G) = \left\lceil\max_{G' \subseteq G} \frac{|E(G')|}{|V(G')| - 1} \right\rceil$. Therefore, there exists $G'\subseteq G$ with $\frac{|E(G')|}{|V(G')|-1} > \text{ar}(G) - 1$. If Maker restricts her playing to $G'$ for as long as there are unclaimed edges in $G'$, she will claim at least $|E(G')|/2$ edges in $G'$. This ensures that $ar(M)\geq \frac{\ell}{2}$, which makes this explicit strategy optimal.
\end{remark}

Still in regards to the $H$-game, it would be interesting to know the order and/or size of the smallest graph $H$ for which the $H$-game is \PSPACE-complete. Further, we only have positive results for the $H$-game played in trees. We wonder for which graphs $H$, the $H$-game in trees can be solved in polynomial time. Also, as the $H$-game is \PSPACE-complete in graphs of diameter at most~$6$, and Maker-Breaker games, in general, are \W[1]-hard parameterized by the length of the game~$k$~\cite{Bonnet2017}, 
is the $H$-game \W[1]-hard parameterized by $k$? 

In general, it would be intriguing to know the complexity of other classic Maker-Breaker games played on edge sets of general graphs, such as the Hamiltonicity game.
Lastly, what about these games for other conventions like Avoider-Enforcer?

\bibliographystyle{plain}
\bibliography{references}

\end{document}